\documentclass[12pt]{article}
\usepackage{geometry, booktabs, subcaption}
\usepackage{titlesec}
\usepackage[T1]{fontenc}
\usepackage[utf8]{inputenc}
\usepackage{authblk}
\usepackage{amssymb,amsmath}
\usepackage{graphics}
\providecommand{\keywords}[1]{\textbf{Keywords: } #1}
\titleformat{\section}{\centering\large\scshape}{\thesection}{1em}{}
\titleformat{\subsection}{\centering\normalsize\scshape}{\thesubsection}{1em}{}

\usepackage{subcaption}
\usepackage[T1]{fontenc}
\usepackage[utf8]{inputenc}
\usepackage{amssymb,amsmath}
\usepackage{graphics}

\usepackage{amsmath,amsthm,amsfonts,latexsym,fullpage,graphicx,vmargin,mathrsfs,xcolor}
\usepackage{subcaption}
\usepackage{multirow}
\usepackage{comment}
\usepackage[inline]{enumitem}
\usepackage{hyperref}
\usepackage{natbib}
\usepackage{graphicx} 

\usepackage{placeins}
\usepackage{dsfont}

\usepackage{xr}
\usepackage{color}
\usepackage{amsmath}
\usepackage{amsfonts,dsfont,mathrsfs}
\usepackage{booktabs,multirow,array}
\usepackage[linesnumbered,ruled,vlined]{algorithm2e}

\makeatletter
\renewcommand{\algocf@captiontext}[2]{#1\algocf@typo. \AlCapFnt{}#2} 
\def\@algocf@capt@plain{top}
\renewcommand{\algocf@makecaption}[2]{%
	\addtolength{\hsize}{\algomargin}%
	\sbox\@tempboxa{\algocf@captiontext{#1}{#2}}%
	\ifdim\wd\@tempboxa >\hsize
	\hskip .5\algomargin%
	\parbox[t]{\hsize}{\algocf@captiontext{#1}{#2}}
	\else%
	\global\@minipagefalse%
	\hbox to\hsize{\box\@tempboxa}
	\fi%
	\addtolength{\hsize}{-\algomargin}%
}
\makeatother

\newtheorem{proposition}{Proposition}
\newtheorem{theorem}{Theorem}
\newtheorem{lemma}{Lemma}

\newcommand{\iid}{\stackrel{\mbox{\scriptsize iid}}{\sim}}
\newcommand{\ind}{\stackrel{\mbox{\scriptsize ind}}{\sim}}
\newcommand{\bm}[1]{\mbox{\boldmath{$#1$}}}

\newcommand{\calN}{\mathcal{N}}

\newcommand{\ptilde}{\widetilde{p}}
\newcommand{\mutilde}{\widetilde{\mu}}
\newcommand{\dd}{\mathrm d}
\newcommand{\Ga}{\mbox{Gamma}}
\newcommand{\Cov}{\mathrm{Cov}}
\newcommand{\Var}{\mathrm{Var}}

\newcommand{\indicator}{\mathds{1}}

\newcommand{\X}{\mathbb{X}}

\newcommand{\E}{\mathbb{E}}
\newcommand{\R}{\mathbb{R}}

\renewcommand{\mid}{\,|\,}

\usepackage[normalem]{ulem}

\title{Normalized Latent Measure Factor Models}

\author[1,2]{Mario Beraha}
\author[3]{Jim Griffin}
\affil[1]{\normalsize{Department of Mathematics, Politecnico di Milano}} 
\affil[2]{\normalsize{Department of Computer Science, Universit\`{a} degli Studi di Bologna}}
\affil[3]{\normalsize{Department of Statistical Science, University College London}}

\begin{document}
\maketitle

\begin{abstract}

We propose a methodology for modeling and comparing probability distributions within a Bayesian nonparametric framework.
Building on dependent normalized random measures, we consider a prior distribution for a collection of discrete random measures where each measure is a linear combination of a set of \emph{latent} measures, interpretable as characteristic traits shared by different distributions,  with positive random weights. 
The model is non-identified and a method for post-processing posterior samples to achieve identified inference is developed. This uses 
Riemannian optimization to solve a non-trivial optimization problem over a Lie group of matrices.
The effectiveness of our approach is validated on simulated data and in two applications to two real-world data sets: school student test scores and personal incomes in California. Our approach leads to interesting insights for populations and easily interpretable posterior inference.
\end{abstract}

\section{Introduction}

Modeling a set of related probability measures is a common task in Bayesian statistics, the most common example being
when covariates are associated with each observation.
In this work, we consider the case of a single discrete-valued covariate, which might be regarded as a group indicator, that is, when data are naturally divided into subpopulations or groups.
{One of the main motivations for these kinds of analyses is combining data from different sources or experiments, where, for each source, a set of observations is collected:  pooling together all the data could ignore important differences across populations while modeling each group separately might result in poor performance especially if the number of observations in each group is small.
Applications range from population genetics \citep{elliot2019hdp} to healthcare \citep{muller2004method,rodriguez2008nested} and text mining \citep{teh2006hierarchical}.

\begin{figure}
\centering
\includegraphics[width=\linewidth]{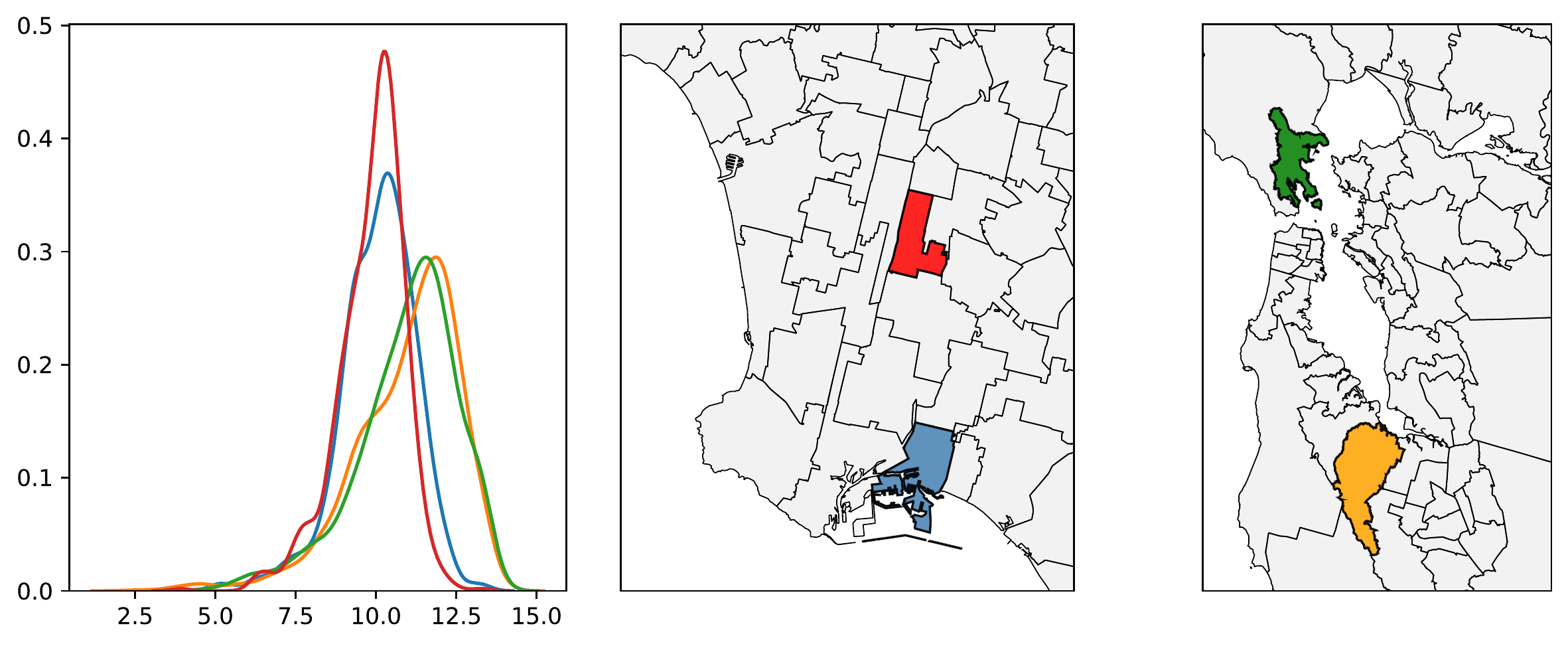}
\caption{Kernel density estimates of the (log) personal incomes in four areas in California (left plot): two in Los Angeles (middle plot) and two in San Francisco (right plot).}
\label{fig:example_income}
\end{figure}

Within this setting, our goal is to propose a flexible model that, in addition to combining heterogeneous sources of data, gives an efficient way of representing the difference in distribution across populations.
Consider for example Figure \ref{fig:example_income}, which displays the distribution of the personal annual income (on the log scale) in four different geographic areas of California: two in Los Angeles and two in San Francisco.
In this case, similarities and differences between the distributions can be easily spotted by eye: the two areas in Los Angeles are associated with (much) lower incomes than the areas in San Francisco.
When the number of groups increases, it is not possible to carry out these comparisons by eye.
Our model provides a way to decompose the area-specific densities into a linear combination of ``common traits'', which are themselves probability measures. 
In Section \ref{sec:income_data}, we provide a thorough analysis of the Californian income data, finding four common traits, associated with an average distribution of income, and a prevalence of low, medium, and high incomes respectively.
By looking at the weights (of the linear combination of common traits) associated with the four groups in Figure \ref{fig:example_income}, we easily spot differences between the Los Angeles and San Francisco areas: the weight associated to the low-income trait is large in the first areas and low in the second two; vice versa for the weight associated to the high-income trait. See Figure \ref{fig:income_map} for more details.

To formalize the discussion above, let $(\bm y_1, \ldots, \bm y_g)$, $\bm y_j = (y_{j1}, \ldots, y_{j n_j})$ denote a sample of observations divided into $g$ groups. 
A common assumption is that data are exchangeable in each group, but exchangeability might not hold across different groups. In particular, by de Finetti's theorem, this is tantamount to assuming that there is a vector of random probability measures $(p_1, \ldots, p_g) \sim Q$ such that, in each group, $y_{j1}, \ldots, y_{jn_j} \iid p_j$ and that independence, conditionally on $p_1, \ldots, p_g$, holds across groups.
We focus here on mixture models of the kind $p_j(y) = \int_{\Theta} f(y \mid \theta) \ptilde_j(\dd \theta)$.

The construction of a flexible prior $Q$ that can suitably model heterogeneity while borrowing information across different groups has been thoroughly studied in Bayesian nonparametrics.
Previously proposed approaches consider constructing $\ptilde_1, \ldots, \ptilde_g$ in a hierarchical model fashion \citep{teh2006hierarchical,camerlenghi2019distribution,bassetti2019hierarchical,argiento2019hcrm,beraha2021semi}, considering convex combinations of shared and group-specific random measures \citep{muller2004method}, starting from additive processes \citep{griffin2013comparing, lijoi2014gm} and nested processes \citep{rodriguez2008nested,camerlenghi_nested}. See \cite{quintana2022dependent} for a recent review.

As previously mentioned, the focus of the present paper is slightly different. First of all, we are interested in the situation when the number of groups $g$ is large relative to the sample size in each group $n_j$. Then, it is likely that the dataset cannot inform the huge number of parameters that are associated with extremely flexible models and we advocate for a more parsimonious model where substantial sharing of information is encouraged across different groups of data.
Moreover, in addition to modeling the densities $\ptilde_1, \ldots, \ptilde_g$, we also want to identify the main differences in distribution of the data across groups.
To the best of our knowledge, this question has not been addressed systematically in the Bayesian nonparametric literature.
In the frequentist one, several approaches to principal component analysis for probability distribution have been proposed, see for instance \cite{pegoraro_beraha_pca} and the references therein.

The setting ``large $g$, small $n_j$'' is somewhat reminiscent of high-dimensional data analysis, where the dimension of each observation is large relative to the sample size. 
In this case, latent factor models \citep[see, e.g., ][]{arminger_factor_models} provide a powerful tool.
In a latent factor model, it is assumed that each observation $x_i \in \R^p$ is a linear combination of a set of $H$ $d$-dimensional latent factors weighted by observation-specific scores, plus an isotropic error term.
We follow this analogy and propose \emph{normalized latent measure factor models}, a class of prior distributions for a vector of random probability measures $\ptilde_1, \ldots, \ptilde_g$.
Informally, our model amounts to considering  $\ptilde_j$ as a convex combination of a set of latent random probability measures, see Section \ref{sec:model}.

Our construction shares similarities with \cite{griffin2013comparing} and \cite{lijoi2014gm}. 
There, the authors assume each $\ptilde_j$ as the normalization of a random measure obtained by superposing several completely random measures. Essentially, this is analogous to our approach if we let all the scores (before some normalization step, see Section \ref{sec:model}) be zero or one.
The main difference is that, since their scores are binary, they usually assume that the number of latent factors $H$ is larger than the number of groups $g$. 
This leads to posterior simulation algorithms that can scale and/or mix poorly with $g$.
Moreover, they do not consider the problem of decomposing the populations' distribution into interpretable common traits, which necessarily requires $H$ to be much smaller than $g$.

As is usually the case for latent factor models, our model is not identifiable, due to two parameter matrices entering multiplicatively.
To tackle this issue, we propose post-processing the MCMC chains to find an ``optimal representative'' for both parameters which leads to a non-trivial optimization problem. 
Indeed, taking into account the invariance to scaling of normalized random measures leads to formulating the optimization over a Riemannian manifold of matrices, specifically the special linear group (matrices whose determinant is equal to one).
Moreover, additional constraints must be taken into account to ensure the positiveness of both parameters and we propose an iterative algorithm based on gradient descent.
The first constraint (determinant equal to one) can be tackled by means of differential geometric tools: leveraging the differential structure of the special linear group, we use a variant of Riemannian gradient descent which ensures that all the intermediate points of the algorithm lie inside the special linear group.
To take into account the positivity constraints, we propose to use the augmented Lagrangian multiplier method within the previously discussed Riemannian framework, leading to a Riemannian augmented Lagrangian multiplier method.

We consider two motivating applications. The first one is the scores on a mathematics test of approximately $40,000$ students in $1048$ Italian high schools from the \emph{invalsi} dataset. The median number of students taking the test in each high school is as little as 37, the minimum being 4 and the maximum 131. 
The second one comes from the US income survey. Here, the groups are represented by geographical units called \emph{PUMAs}, which correspond to areas with roughly $100,000$ inhabitants. 
We show how our model can be adapted to induce correlation between PUMAs that are geographically close, by assuming that the scores are distributed as a log Gaussian Markov random field. 
Compared to traditional spatial factor models, we introduce the spatial dependence in the loadings matrix instead of the latent factors.

The rest of the paper is organized as follows. Section \ref{sec:model} formalizes our model and discusses its statistical properties.
Section \ref{sec:algo} describes the MCMC algorithm for posterior inference and we present our post-processing algorithm in Section \ref{sec:postproc}. Section \ref{sec:simu} and Section \ref{sec:real_data} present numerical illustration on simulated data and real data, respectively.
Finally, we discuss possible extensions of the proposed approach in Section \ref{sec:discussion}.
The Appendix collects background material on Riemannian optimization and completely random measures, proofs of the theoretical results, and additional simulations.
\texttt{Python} code implementing the MCMC and the post-processing algorithms is available at \url{github.com/mberaha/nrmifactors}.

\section{The Model}\label{sec:model}

For simplicity and specifity, we assume that each $y_{ji} \in \R^d$ and that $\Theta \subset \R^q$ for some $d$ and  $q$. 
The results can be easily extended to the case when $y_{ji}$ are elements of a complete and separable (i.e., Polish) metric space and $\Theta$ is Polish as well.

To keep the discussion light, we defer all technical details and the proofs of the results to the Appendix.

\subsection{Preliminaries}

Before presenting our model in detail, we give some background material on completely random measure and their normalization. This will constitute the backbone of our approach.

Let $(\Theta, \mathcal B(\Theta))$ be a complete and separable metric space endowed with its Borel $\sigma$-algebra.
A random measure is a random element $\mu$ taking values in the space of probability measures over $\Theta$, such that $\mu(B) < +\infty$ almost surely for all $B \in  \mathcal B(\Theta)$.
Such a measure is termed completely random by \cite{Kin67} if, for pairwise disjoint $B_1, \ldots, B_n \in \mathcal B(\Theta)$, the random variables $\mu(B_j)$, $j=1, \ldots, n$, are independent.
For our purposes, it is sufficient to consider completely random measures of the kind $\mu(A) = \int_{ \mathbb R_+ \times A} s N(\dd s\, \dd x)$,
where $N$ is a Poisson point process on $\Theta \times \mathbb R_+$ with base (intensity) measure $\rho(\dd s\, \dd x)$. We further assume  $\rho(\dd s \,\dd x) = \nu(\dd s)\, \alpha(\dd x)$ where $\nu$ is a L\'evy measure on the positive reals, $\alpha$ is a Borel measure on $\Theta$. See, e.g., \cite{Kin93}  for a detailed account of random measures. 

A fruitful approach to constructing random probability measures is by normalization of completely random measures, i.e., by setting $p(\cdot) = \mu(\cdot) / \mu(\Theta)$, which was originally introduced in \cite{ReLiPr03}.
For the random measure $p$ to be well defined, one must ensure that $\mu(\Theta) > 0$ and $\mu(\Theta) < +\infty$ almost surely. As shown in \cite{ReLiPr03}, sufficient conditions are 
$\int_{\mathbb R_+} \nu(\dd s)  = +\infty$ and $\int_{\mathbb R_+} \min\{1, s\} \,\nu(\dd s) < +\infty$.

\subsection{Normalized Latent Measure Factor Models}

As already mentioned in the Introduction, we assume
\[
	y_{j1}, \ldots, y_{jn_j} \mid \ptilde_j \iid p_j := \int_{\Theta} f(\cdot \mid \theta) \ptilde_j(\dd \theta)
\]
and that each $\ptilde_j$ is a normalized random measure, that is 
\[
    \ptilde_j(\cdot) = \frac{\mutilde_j(\cdot)}{\mutilde(\Theta)}, \qquad j=1, \ldots, g.
\]
Then, the model is specified by a choice of the mixture kernel $f(\cdot \mid \cdot)$ and a prior distribution for $(\mutilde_1, \ldots, \mutilde_g)$.
Let $\mu^*_1, \ldots, \mu^*_H$ be a completely random vector (i.e., a vector of completely random measures).
Let $\lambda_{jh}$, $j=1, \ldots, g$, $h=1, \ldots, H$ be a double sequence of almost surely positive random variables (specific choices of the distribution of the $\lambda_{jh}$'s are discussed later).
We assume 
\begin{equation}\label{eq:lat_fac_rm}
    \mutilde_j(\cdot) = \sum_{h=1}^H \lambda_{jh}\, \mu^*_h(\cdot).
\end{equation}

We could  choose  $(\mu^*_1, \ldots, \mu^*_H)$ to be independent and identically distributed random measures, i.e.
\[
    \mu^*_h(\cdot) = \sum_{k \geq 1} W_{hk}\, \delta_{\theta^*_{hk}}(\cdot)
\]
where $\{W_{hk}, \theta^*_{hk}\}_{k=1}^\infty$ are the points of a Poisson point process on $[0, +\infty) \times \Theta$ with, for instance, intensity $\nu_h(\dd s_h\, \dd x_h) = \rho(s_h) \dd s_h\, \alpha(\dd x_h)$, i.e., all the intensities are equal.
This choice leads to a particularly tractable model for $(\mutilde_1, \ldots, \mutilde_g)$ as we have that marginally, each $\mutilde_j$ is a completely random measure as specified in the following proposition.
\begin{proposition}\label{prop:levy_crm}
    Let $\mutilde_j = \sum_{h=1}^H \lambda_{jh} \mu^*_h$ where the $\mu^*_h$'s are completely random measures with associated L\'evy intensity $\nu^*_h(\dd s_h, \dd x_h) = \rho^*_h(s_h) \dd s_h \, \alpha^*_h(\dd x_h)$. Further, assume that the $\mu^*_h$'s are independent. Then $\mutilde_j$ is a completely random measure with L\'evy intensity
    \[
        \nu_j(\dd s, \dd x) = \sum_{h=1}^H \frac{1}{\lambda_{jh}} \rho^*_h(s / \lambda_{jh}) \alpha^*_h(\dd x)
    \]
\end{proposition}

We find that a more suitable model for our applications arises when $\mu^*_1, \ldots, \mu^*_H$ share their support points.
In particular, we will assume that $\mu^*_1, \ldots, \mu^*_H$ is a compound random measure \citep[CoRM,][]{griffin2017compound}.
That is, 
\[
    \mu^*_h(\cdot) = \sum_{k \geq 1} m_{hk} J_k \delta_{\theta^*_k}(\cdot),
\]
where $m_{hk}$ are positive random variables such that $m_k = (m_{1k}, \ldots, m_{Hk})$, $k \geq 1$, are independent and identically distributed from a probability measure on $\mathbb R_+^H$, and $\eta = \sum_{k \geq 1} J_k \delta_{\theta^*_k}$ is a completely random measure with L\'evy intensity $\nu^*(\dd z) \alpha(\dd x)$.
We argue that a CoRM-based construction should be preferred to an independent CRMs-based one since
(i) sharing atoms across all measures is linked to better predictive performance  \citep{quintana2022dependent}, (ii)  the number of parameters involved is much smaller, which ultimately leads to the possibility of fitting this model to large datasets, and (iii) each latent factor $\mu^*_h$ can be interpreted separately (through the post-processing algorithm presented in Section \ref{sec:postproc}).
The effectiveness of this model comes with a tradeoff in analytical tractability, since, as shown in the Appendix, the random measure \eqref{eq:lat_corm} is not completely random.

In this case we can write
\begin{equation}\label{eq:lat_corm}
    \mutilde_j(\cdot) = \sum_{k \geq 1} (\Lambda M)_{jk} J_{k} \delta_{\theta^*_k}(\cdot),
\end{equation}
where $\Lambda$ is the $J \times H$ matrix with entries $\lambda_{jh}$, $M$ is a $H \times \infty$ matrix, so that $\Gamma = \Lambda M$ is a $g \times \infty$ matrix with entries $\gamma_{jk}$, 
 $j = 1, \ldots, g$, $k \geq 1$.
Note that, in analogy to CoRMs, also our model includes shared weights $J_k$ for all the measures $\mutilde_j$.
We find that the additional borrowing of strength obtained through the $J_k$'s is useful in practice since, in our applications, the $\mutilde_j$'s are usually similar.

Equations \eqref{eq:lat_fac_rm} and \eqref{eq:lat_corm} share analogies to latent factor models, where  the observed variable is $X \in \mathbb{R}^p$ and its $\ell$-th entry is modeled as $X_\ell \approx \sum_{h=1}^H  \omega_{\ell h} Z_h$, for $Z=(Z_1, \ldots, Z_H)$ an $H$-dimensional random variable.
In particular, we could consider $\mu^*_1, \ldots, \mu^*_H$ to be measure-valued factor loadings and the $\lambda_{jh}$'s to be factor scores. 
This yields an interpretation similar to functional factor models \citep{montagna2012bayesian}. On the other hand, we could consider the measure-valued vector $(\mutilde_1, \ldots, \mutilde_g)$ as a single high-dimensional observation, and model it as a linear combination of measure-valued factors with loadings $\lambda_{jh}$'s.
Both interpretations make sense and lead to interesting analogies. We use the latter one and call $\Lambda$ the loadings matrix and the $\mu^*_h$'s the latent measures.

Prior elicitation is required to set the L\'evy intensity $\nu^*$ of the CoRM, the distribution of the scores $m_{hk}$, and the distribution of $\Lambda$.
Following \cite{griffin2017compound}, we assume that $m_{hk} \iid \mbox{Ga}(\phi)$, where $\mbox{Ga}(\phi)$ denotes the law of a gamma random variable with shape parameter $\phi$ and rate parameter 1(we will also use $\mbox{Ga}(\phi, \beta)$ to denote a gamma random variable with rate parameter $\beta\neq 1$). 
Therefore, the dependence across the $\mutilde_j$'s depends on $H$, $\nu^*$, and $\Lambda$.

The prior for $\Lambda$ allows us to address several interesting modeling questions.
When no additional group-specific information is available, such as comparing the distribution of test results in different schools, a natural choice would be to assume the $\lambda_{ij}$'s i.i.d. from some probability distribution with support on $\R_+$, such as the gamma distribution.
We find it more convenient to specify a \emph{shrinkage} prior on $\Lambda$, to automatically select the number of latent factors $H$.
This approach has received considerable attention in Gaussian latent factor models, see, for instance, \cite{BhDu11, LeDuDu20, schiavon_infinite_factor}.
In our example, we consider $\Lambda$ distributed as a multiplicative gamma process  \citep{BhDu11},
\begin{equation}\label{eq:lambda_mgp}
	\lambda_{jh} = (\phi_{jh} \tau_h)^{-1}, \ \tau_{h} = \prod_{j=1}^h \theta_j, \ \theta_1 \sim \mbox{Ga}(a_1), \ \theta_2, \ldots \iid  \mbox{Ga}(a_2), \ \phi_{jh} \iid  \mbox{Ga}(\nu/2, \nu/2).
\end{equation}
In Section \ref{sec:algo} we propose a variant of the adaptive Gibbs sampler of \cite{BhDu11} to automatically select $H$ in the first iterations of the MCMC algorithm.

If group-specific information, such as covariates, is available, we can model the finite-dimensional matrix $\Lambda$. For example, the PUMAs in the Californian income data are indexed by a specific areal location.
This can be modelled using a  $g \times g$ spatial proximity matrix denoted by $W$, where $W_{j \ell} = 1$ if areas $j$ and $\ell$ share an edge and $W_{j \ell} = 0$ otherwise, but more general choices of proximity could be considered in other examples.
Then, we can encourage spatial dependence between the $\mutilde_j$'s by assuming 
\begin{equation}\label{eq:lambda_gmrf}
    \log \bm \lambda^h \iid \mathcal{N}_H\left(\mu, \left(\tau (F - \rho W)\right)^{-1}\right), \qquad h=1, \ldots, H
\end{equation}
where $\bm \lambda^h = (\lambda_{1h}, \ldots, \lambda_{gh})$ is the $h$--th column of the matrix $\Lambda$, $F$ is a diagonal matrix with entries $F_{ii} = \sum_j W_{ij}$, and $\rho \in (0, 1)$. 
We suggest setting $\mu = \log(1/H, \ldots, 1/H)$ in \eqref{eq:lambda_gmrf} to encourage a priori each $\mutilde_j$ to be a convex combination of the $\mu^*_h$'s with equal weights.
The model could also be applied to  geo-referenced data using a log Gaussian process,
\[
    \log \bm \lambda^h \iid \mathcal{GP}(\mu, \mathcal{K}), \qquad h=1, \ldots, H
\]
where $\bm \lambda^h = (\lambda_{1h}, \ldots, \lambda_{gh})$ is the $h$--th column of the matrix $\Lambda$.


\subsection{Some Statistical Properties}

In this section, we discuss some distributional properties of the measures $\mutilde_j$'s in light of the prior assumption above.
We assume that the $\lambda_{jh}$'s are independent of $\mu^*_1, \ldots, \mu^*_H$.
Firstly, it is clear that
\[
    \E[\mutilde(A)] = \sum_{h=1}^H \E[\lambda_{jh}] \E[\mu^*_h(A)].
\]
When we consider the normalized measures, the expression of the expected value is more complex.
\begin{theorem}\label{teo:expectation}
    Let $(\mu^*_1, \ldots, \mu^*_H)$ be a CoRM with i.i.d. scores. 
        Denote the Laplace transform of the scores' distribution by $\mathcal{L}_m(u) := \E[e^{-u m}]$  and let $\kappa_m(u, n) := \E[e^{-u m} m^n]$.
    Then for all measurable $A \subset \Theta$
    \begin{multline*}
        \E[\ptilde_j(A)] = \\ \alpha(A) \sum_{h=1}^H \int \E\left[ \lambda_{jh}  \psi_{\rho}(u\lambda_{j1}, \ldots, u \lambda_{jH}) \int_{\R_+} z \prod_{k \neq h} \mathcal{L}_m(u \lambda_{jk} z) \kappa_m(u \lambda_{jh} z, 1) \nu^*(\dd z) \right] \dd u 
    \end{multline*}
    where $\psi_{\rho}$ is the Laplace functional of $(\mu^*_1, \ldots, \mu^*_H)$ (evaluated at the constant functions $u\lambda_{j1}, \ldots, u \lambda_{jH}$).
\end{theorem}
Although it is not possible to evaluate the quantity in Theorem \ref{teo:expectation} analytically, a priori Monte Carlo simulation can be used to numerically estimate the expected value of $\ptilde_j(A)$.

To characterize the dependence induced by the latent measure factor model, an intuitive measure is the covariance between two random measures.
\begin{proposition}\label{prop:cov_mu}
The following expression holds.
    \begin{multline}\label{eq:cov_mu}
        \Cov\left[\mutilde_j(A), \mutilde_\ell(B)\right]  =  \\ \sum_{h, k} \E[\lambda_{jh} \lambda_{\ell k}] \Cov (\mu^*_h(A), \mu^*_k(B)) + \Cov(\lambda_{jh}, \lambda_{\ell k}) \E[\mu^*_h(A) \mu^*_k(B)] 
    \end{multline}
    If the $\lambda_{jh}$'s have the same marginal distribution, the $\mu^*_h$'s have the same marginal distribution, $\lambda_j = (\lambda_{j1}, \ldots, \lambda_{jH})$ and $\lambda_\ell$ (defined analogously) are independent, $\E[\lambda_{jh} \lambda_{\ell h}] = \kappa$,  $\Cov(\lambda_{jh}, \lambda_{\ell h}) = \rho$ for all $j, \ell, h$, then:
    \begin{multline*}
	\Cov\left[\mutilde_j(A), \mutilde_\ell(B)\right]  =  \\ \Cov(\mu^*_1(A), \mu^*_1(B)) \kappa H + m^*_1(A) m^*_1(B) \rho H + \sum_{h \neq q} \bar \lambda_{11}^2  \Cov(\mu^*_h(A), \mu^*_k(B))
\end{multline*}
    where $\bar{\lambda}_{jh} := \E[\lambda_{jh}]$ and  $m^*_h(A) = \E[\mu^*_h(A)]$.

Finally, if in addition the $\mu^*_h$'s are independent, the latter sum disappears
\end{proposition}
From \eqref{eq:cov_mu}, it is clear that $\Cov\left[\mutilde_j(A), \mutilde_\ell(B)\right]$ increases with: (i) the correlation of the measures at the latent lavel ($\Cov (\mu^*_h(A), \mu^*_k(B))$ large), (ii) the correlation of the scores ($\Cov(\lambda_{jh}, \lambda_{\ell k})$ large), (iii) large values in the scores $(\E[\lambda_{jh} \lambda_{\ell k}]$ large), (iv) random measures with large masses ($\E[\mu^*_h(A), \mu^*_k(B)]$ large), and (v) large values of $H$ (more terms in the summation).

The correlation between $\mutilde_j(A)$ and $\mutilde_\ell(B)$ can be formally derived from \eqref{eq:cov_mu} but its expression is not easily interpretable in general. 
To get a nicer expression, assume $A=B$, $\Cov(\mu^*_h(A), \mu^*_k(A)) = \Cov(\mu^*_{m}(A), \mu^*_{n}(A)) = c_{A}$, $\E[\mu^*_h(A)] = \E[\mu^*_k(A)] = m_A$.
Then 
\begin{multline*}
	\Cov\left[\mutilde_j(A), \mutilde_\ell(A)\right] = \E[\mu^*_1(A)^2] \left( \sum_{h=1}^H \E[\lambda_{jh} \lambda_{\ell h}] \right) + \\ (c_A + m_A^2) \left( \sum_{h \neq k} \E[\lambda_{jh} \lambda_{\ell k}] \right) - m_A^2 \left( \sum_{h, k} \bar{\lambda}_{jh} \bar{\lambda}_{\ell k}\right)
\end{multline*}
Let us specialize the above expression further.
Consider first the case of independent scores $\lambda_{jh} \iid \mbox{Ga}(\psi, 1)$.
The correlation between $\mutilde_j(A)$ and $\mutilde_\ell(A)$ amounts to
\begin{equation}\label{eq:corr_iid_scores}
		\left( 1 + \frac{m_A}{\left(\Var[\mu^*_1(A)] + c_A(H-1) \right) \psi } \right)^{-1}
\end{equation}
which is an increasing function of $H$ and $\psi$ as expected. See Appendix \ref{sec:proofs} for a proof.

To evaluate $m_A$, and $c_A$ we use the following result.
\begin{proposition}\label{prop:moments_corm}
	Consider a CoRM with $\mbox{Ga}(\phi)$ distributed scores and gamma process marginals (i.e., each $\mu^*_h$ is distributed as a gamma process).
	Then for any measurable $A$:
	\begin{enumerate}
		\item $\E[\mu^*_h(A)] = \alpha(A)$,
		\item $\E[\mu^*_h(A)\mu^*_k(A)] = (\alpha(A) + \alpha(A)^2) \phi^2 (B(1, \phi))^2 3/2$, where $B(a, b)$ denotes the Beta function.
	\end{enumerate}
\end{proposition}

Consider now the case when $\Lambda$ is distributed as a multiplicative gamma process introduced in \cite{BhDu11}.
In this case, we don't have an interpretable expression for the correlation between $\mutilde_j(A)$ and $\mutilde_\ell(A)$.
In the Appendix \ref{sec:proofs} we report the expressions for $\Cov\left[\mutilde_j(A), \mutilde_\ell(A)\right]$ and $\Var[\mutilde_j(A)]$ which might be used to numerically compute the desired correlation.
Figure \ref{fig:corr_mgp} displays the correlation between $\mutilde_j(A)$ and $\mutilde_\ell(A)$ for a set $A$ such that $\alpha(A) = 0.5$.
We notice that when the CoRM has gamma process marginals, the parameter $\phi$ has little effect on the correlation between the $\mutilde_j$'s.
On the contrary, there is a strong interaction between $a_2$, $\nu$, and $H$. 
For smaller values of $\nu$, larger values of $H$ imply a higher correlation. 
When $\nu$ is sufficiently large (e.g. larger than 6), the effect of $H$ is less evident. Moreover, larger values of $a_2$ imply a weaker correlation. 
This is expected as it essentially reduces the number of active latent measures.
In Figure \ref{fig:corr_lgmrf} in the Appendix, we show the correlation between $\mutilde_j(A)$ and $\mutilde_\ell(A)$  under   prior \eqref{eq:lambda_gmrf}  for different choices of areas $j$ and $\ell$, as a function fo $\tau$ and $\rho$.

\begin{figure}[t]
\centering
\includegraphics[width=\linewidth]{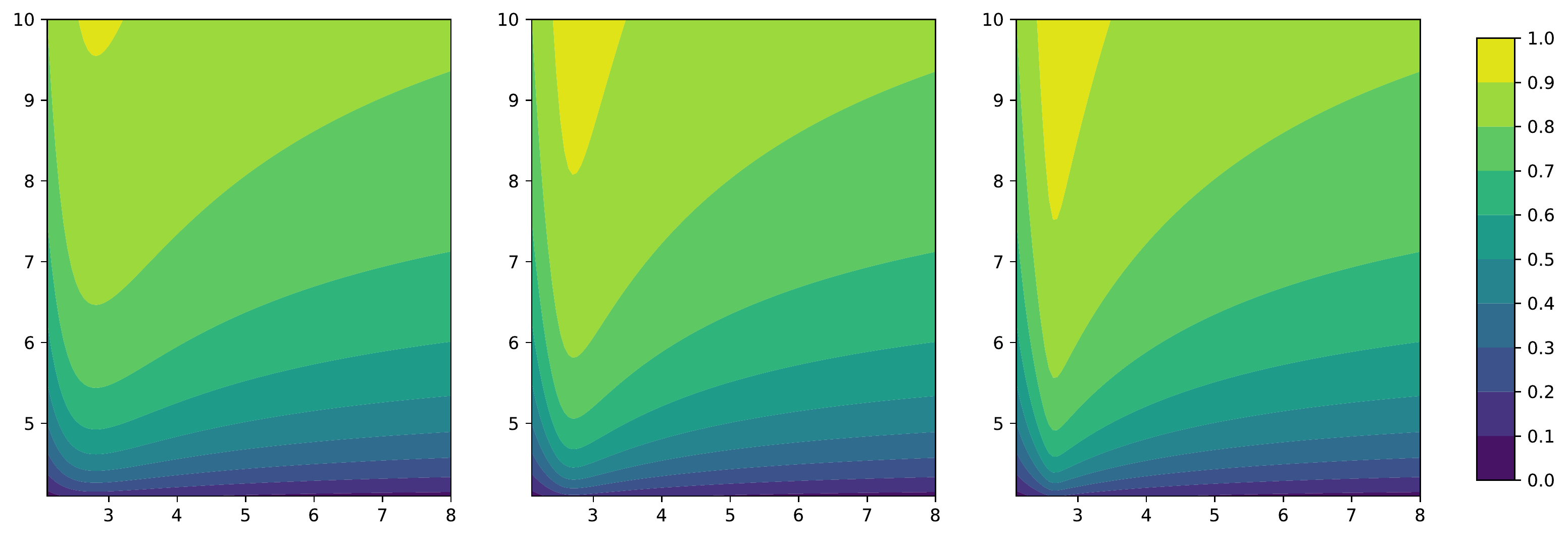}
\caption{Correlation between $\mutilde_j(A)$ and $\mutilde_\ell(A)$ for a set $A$ such that $\alpha(A) = 0.5$, under the multiplicative gamma process prior. $a_1 = 2.5$, $\phi = 2$. From left to right $H=4, 8, 16$. The values of $a_2$ vary across the $x$-axis in each plot, the values of $\nu$ across the $y$-axis.}
\label{fig:corr_mgp}
\end{figure}

Since the atoms are shared across all the measures $\mutilde_j$'s, another possible way of characterizing the dependence between two measures is to consider the ratio of weights associated to the $k$--th atom in $\mutilde_j$ and $\mutilde_\ell$,
\begin{equation}\label{eq:jump_ratio}
    r^k_{j \ell} := \frac{(\Lambda M)_{j k}}{(\Lambda M)_{\ell k}} = \frac{\sum_{h=1}^H \lambda_{j h} m_{h k}}{\sum_{h=1}^H \lambda_{\ell h} m_{h k}}.
\end{equation}
A trivial upper bound is
\[
    r^k_{j \ell} \leq \sum_{h=1}^H \frac{\lambda_{jh}}{\lambda_{\ell h}}.
\]
Multiplying and dividing by $H$ in \eqref{eq:jump_ratio} and taking the logarithm yields
\[
    \log r^k_{j \ell} =  \log \left( \frac{1}{H}\sum_{h=1}^H \lambda_{j h} m_{h k} \right)- \log \left( \frac{1}{H}\sum_{h=1}^H \lambda_{\ell h} m_{h k} \right).
\]
By the strong law of large numbers, we have that $\log r^k_{j \ell} \rightarrow 0$ as $H \rightarrow \infty$ if, for instance, $\lambda_{jh}$ and $\lambda_{\ell h}$ are independent and identically distributed across the values of $h$.
Moreover, it is clear that the variance of $r^k_{j \ell}$ increases with the variance of the $\lambda_{jh}$'s.
In Appendix \ref{sec:app_simu} we report an a prior Monte Carlo simulation comparing $r_{j \ell}$ as a function of $H$ under different priors for $\Lambda$, namely and i.i.d. prior with $\mbox{Ga}(\psi)$ distributed $\lambda_{jh}$'s, the multiplicative gamma process in \eqref{eq:lambda_mgp} and the the cumulative shrinkage prior \cite{LeDuDu20}.
It is clear that under the two latter shrinkage priors, the choice of $H$ has a smaller impact on the prior.
For the sake of computational efficiency, we will adopt the multiplicative gamma process prior in our simulations, when no additional group-specific covariates are present.
Instead, when we consider the case of area-referenced groups, we consider $H$ to be a hyperparameter and perform model selection based on predictive performance

\section{Posterior Inference}\label{sec:algo}

Let $\alpha$ be a measure on $\Theta$, $\nu^*$ a L\'evy intensity on $\R_+$, and $\phi > 0$. We denote with $\mbox{CoRM}(\phi, \nu^*, \alpha)$ the law of a compound random measure with i.i.d. $\mbox{Ga}(\phi)$-distributed scores with directing random measure with intensity $\nu^*(z) \dd z \, \alpha(\dd \theta)$.
Our model can be compactly summarized as 
\begin{equation}\label{eq:model}
    \begin{aligned}
        y_{ji} \mid \theta_{ji} &\ind k(\cdot \mid \theta_{ji}), \qquad i=1, \ldots, n_i \\
        \theta_{ji} \mid \mutilde_j & \iid \mutilde_j / \mutilde_j(\Theta), \qquad i=1, \ldots, n_i \\
        \mutilde_j &:= \sum_{h=1}^H \lambda_{jh} \mu^*_h, \\
        (\mu^*_1, \ldots, \mu^*_h)  \sim &\mbox{CoRM}(\phi, \nu^*, \alpha), \qquad
        \Lambda \sim \pi(\Lambda)  
    \end{aligned}
\end{equation}

In this section, we describe a simple MCMC scheme based on a truncation of the random measures.
In particular, let $K > 0$ denote a fixed number of atoms, we set 
\[
    \mu^*_h = \sum_{k=1}^K m_{hk} J_k \delta_{\theta^*_k}
\]
where $J_k \iid p_J$, with $p_J$ being a probability distribution, and $\theta^*_k \iid G_0 := \alpha / \alpha(\Theta)$. \cite{CaHiHoBr19} provide a thorough review of truncation methods for completely random measures including the choice of $p_J$ for different random measures. 
We use $p_J = \mbox{Beta}(\phi/K, \phi)$ so that $\sum_{k=1}^K J_k  \delta_{\theta^*_k}$ converges to a Beta process as $K \rightarrow +\infty$. This combined with gamma-distributed $m_{hk}$ imply that marginally $\mu^*_h$ follows a gamma process \citep[see][]{griffin2017compound}.
Although this simple truncation might result in an approximation error that is large a priori, as shown in \cite{nguyen2020independent}, posterior inference is usually robust and no significant difference is detected.
The choice of fixing $K$ also allows for (much) faster code since the number of parameters is now fixed, and our implementation can thus take advantage of modern parallelization and vectorization algorithms. 
This is in line with our ultimate goal of fitting very large datasets with our model.
In Appendix \ref{sec:app_slice} we also describe a slice sampling algorithm based on \cite{griffin_slice} that does not require truncating the random measure.

\subsection{MCMC Algorithm for the Truncated Model}

Observe that in \eqref{eq:model}, $\theta_{ji} = \theta^*_k$ with positive probability. Therefore an alternative representation is achieved by introducing latent cluster indicator variables $c_{ji}$ such that $c_{ji}$ are independent categorical variables with support $\{1, \ldots, K\}$ and 
\[
    P(c_{ji} = k \mid \{\lambda_{jh}\}, \{m_{hk}\}, \{J_k\}) \propto (\Lambda M)_{jk} J_k.
\]
Let $T_j := \sum_k (\Lambda M)_{jk} J_k$. Writing $p(\cdot \mid \cdot)$ for a generic conditional density, the joint distribution of data and parameters under \eqref{eq:model} is then
\begin{align*}
    & p(\{y_{j, i}\}, \{c_{j,i}\}, \{\lambda_{j,h}\}, \{m_{h, k}\}, \{J_\ell\}, \{\theta^*_\ell\}) =  \\
    & \qquad \prod_{j=1}^g T_j^{-n_j} \prod_{i=1}^{n_j} f(y_{j, i} \mid \theta^*_{c_{j, i}}) (\Lambda M)_{j, c_{j, i}} J_{c_{j, i}} \times \prod_{h=1}^K \left[G_0(\theta^*_h)  p_{J}(J_k) \prod_{k=1}^K \mbox{Ga}(m_{hk} \mid \phi)  \right] \pi(\Lambda)
\end{align*}
To facilitate posterior inference, we introduce a set of auxiliary variables $u_j$, which are gamma distributed with shape parameter $T_j$ and rate parameter $n_j$. Then  
\begin{align*}
    & p(\{y_{j, i}\}, \{c_{j,i}\}, \{\lambda_{j,h}\}, \{m_{h, k}\}, \{J_\ell\}, \{\theta^*_\ell\}, \{u_j\} ) =  \\
    & \qquad \prod_{j=1}^g  \frac{1}{\Gamma(n_j)} u_j^{n_j - 1} \prod_{i=1}^{n_j} f(y_{j, i} \mid \theta^*_{c_{j, i}}) (\Lambda M)_{j, c_{j, i}} J_{c_{j, i}} \times 
    \exp \left( - \sum_{j=1}^g u_j \sum_{\ell=1}^K (\Lambda M)_{j, \ell} J_\ell \right) \\
    & \qquad \prod_{h=1}^K \left[G_0(\theta^*_h)  p_{J}(J_k) \prod_{k=1}^K \mbox{Ga}(m_{hk} \mid \phi) \right] \pi(\Lambda)
\end{align*}
It is then possible to sample from the posterior distribution via a Gibbs sampler:
\begin{enumerate}
    \item Update the atoms from 
    \[
        p(\theta^*_h \mid \cdots) \propto \prod_{j=1}^g \prod_{i: c_{j, i} = h} f(y_{j, i} \mid \theta^*_h) G_0(\theta^*_h)
    \]
    \item Update the $J$'s from 
    \[
        p(J_\ell \mid \cdots) \propto J_\ell^{q_\ell} \exp \left( - \sum_{j=1}^g u_j (\Lambda M)_{j, \ell} J_\ell \right) p_J(J_{\ell})
    \]
    where $q_\ell = \sum_{j=1}^g \sum_{i=1}^{n_j} I[c_{j, i} = h]$. 
    \item Update the $m$'s from
    \[
        p(M \mid \cdots) \propto \prod_{j=1}^g \prod_{\ell = 1}^K (\Lambda M)_{j, \ell}^{q_\ell} \times \exp \left( - \sum_{j=1}^g u_j (\Lambda M)_{j, \ell} J_\ell \right) \times \prod_{h=1}^H \prod_{k=1}^K \mbox{Ga}(m_{hk} \mid \phi)
    \]
    The update of $M$ can be done in a single block via Hamiltonian Monte Carlo.
    \item Update the $\lambda$'s from
    \[
        p(\Lambda \mid \cdots) \propto \prod_{j=1}^g \prod_{\ell = 1}^K (\Lambda M)_{j, \ell}^{q_\ell} \times \exp \left( - \sum_{j=1}^g u_j (\Lambda M)_{j, \ell} J_\ell \right)  \pi(\Lambda)
    \]
    Again, we can update $\Lambda$ using a single step of Hamiltonian Monte Carlo.
    \item Update the cluster indicators from a categorical distribution over $\{1, \ldots, K\}$ with weights
    \[
        P(c_{j, i} = h \mid \cdots) \propto f(y_{j, i} \mid \theta^*_h) (\Lambda M)_{j, h} J_h
    \]
    \item update the $u$'s from $u_j \mid \cdots \sim \Ga(n_j, T_j)$
\end{enumerate}

Finally, when the prior for $\Lambda$ is the multiplicative gamma process \eqref{eq:lambda_mgp} we propose to gain computational efficiency by selecting $H$ through an adaptive Gibbs sampling scheme as in \cite{BhDu11}.
In particular, when adaptation occurs, we look at the ``empty columns'' of $\Lambda$.
We define a column $h$ of $\Lambda$ to be empty if
\[
	\sum_{j=1}^g \frac{\lambda_{jh}}{\sum_{k=1}^H \lambda_{jk}} < \varepsilon \bar \lambda
\]
where $\bar \lambda = H^{-1} \sum_{h=1}^H  \sum_{j=1}^g \frac{\lambda_{jh}}{\sum_{k=1}^H \lambda_{jk}}$. In our experience $\varepsilon = 0.05$ provides satisfactory results. 
If there are no empty columns, we add a column sampled from the prior to $\Lambda$ and a row sampled from the prior to $M$. Instead, if empty columns are found, we drop them from $\Lambda$ and the corresponding rows from $M$.

\cite{BhDu11} propose to adapt $\Lambda$ at each iteration $\ell$ with a probability $p_\ell$ that decreases exponentially fast.
This choice is possible also within our algorithm but, in our experience, it significantly impacts run-time.
This is due to the choice of using HMC to sample $\Lambda$ and $M$ and, in particular, to the use of the \texttt{tensorflow-probability} Python package, in combination with \texttt{LAX} compilation.
For technical reasons, every time the size of $\Lambda$ and $M$ change, big chunks of the code must be recompiled, so that it's not efficient to adapt every few iterations. Instead, we propose to have a fixed adaptation window of $1,000$ iterations, where the adaptation occurs every $50$ iterations.
In our experience, this simple modification reduces the overall runtime by at least one order of magnitude.

\section{Resolving the non-identifiability via post-processing}\label{sec:postproc}

As already mentioned in the introduction, our model is not identifiable due to the multiplicative relation between $\Lambda$ and $(\mu^*_1, \ldots, \mu^*_h)$.
This is not surprising, as the same holds for common latent factor models \citep{geweke_identifiability}, where the likelihood is invariant to the action of orthogonal matrices.
In that context, a common practice to recover identifiability is to constrain the matrix $\Lambda$ to be lower triangular with positive entries on the diagonal \citep{geweke_lt}.
More recently, it has been proposed to ignore the identifiability issue and obtain a point-estimate of the posterior distribution either by post-processing the MCMC chains \citep[see][and the references therein]{Papastamoulis_2022, dunson_post_proc} or by choosing the maximum a posteriori \citep{schiavon_infinite_factor}.
In particular, \cite{dunson_post_proc} propose to orthogonalize each posterior sample of $\Lambda$ and then solve the sign ambiguity and label switching via a greedy matching algorithm.

The non-identifiability in our model is more severe than the one of common latent factor models. In fact, for any $Q$ s.t. $Q^{-1}$ is well defined, the likelihood is invariant when considering $\Lambda^\prime = \Lambda Q^{-1}$ and $M^\prime = Q M$. 
Nonetheless, the constraints that $\Lambda^\prime \geq 0$ (element-wise) and $M^\prime \geq 0$ greatly reduce the number of matrices $Q$ that can cause non-identifiability.
In particular, we don't need to worry about sign ambiguity. 

\subsection{The Objective Function}

Consider equation \eqref{eq:lat_corm}.
Factorizations of the kind $\Gamma = \Lambda M$ where all the three matrices have nonnegative entries are common in blind source separation (BSS) problems,  where the goal is to estimate ``source components'' $M$ and ``mixing proportions'' $\Lambda$ such that the observed signal $\Gamma$ is approximately $\Lambda M$.
Two well-established approaches to BSS are nonnegative matrix factorization \citep[NMF,][]{sra2005nnmf}  and independent component analysis \citep[ICA,][]{hyvarinen2013ica}. 
The main difference between the two consists in the loss function optimized. In NMF it is usually the norm of the approximation error, while, in ICA, the mutual information between the source components is minimized alongside the approximation error. 
This takes into account the goal of separating the components.
Since in our analogy the sample size of the latent factor model is just one (i.e., in our model there is one single vector $\mutilde_1, \ldots, \mutilde_p$ instead of multiple realizations), it is not possible to use the same criteria of ICA to define what we mean by ``separated components''.
Hence, we propose to optimize with respect to the following \emph{interpretability} criterion: 
\begin{equation}\label{eq:postproc_obj}
	L(Q; M, J, \theta) = \sum_{i < j} \left( \int_{\mathbb{Y}}  \left[ \int_\Theta f(y \mid \theta) \mu^\prime_i(\dd \theta) \right] \left[ \int_\Theta f(y \mid \theta) \mu^\prime_j(\dd \theta) \right] \dd y \right)^2 .
\end{equation}
where 
\[
	 \mu^\prime_j = \sum_{k=1}^K (QM)_{jk} J_k \delta_{\theta^*_k}
\]
Intuitively, low values of $L(Q; M, J, \theta)$ in \eqref{eq:postproc_obj} are attained when the transformed random measures $\mu^\prime_h$, mixed with the mixture kernel $f$, result in well separated densities.

Defining $g_i(y) : = \int_\Theta f(y \mid \theta) \mu^\prime_i(\dd \theta)$ it is clear that \eqref{eq:postproc_obj} can be interpreted as the sum of the squared inner products (in the $L_2$ sense) between $g_i$ and $g_j$.
The $L_2$ distance is not commonly used to measure the discrepancy of densities. A more familiar option would be to consider $\int \sqrt{g_i(y)} \sqrt{g_j(y)} \dd y$, that is $1 - d_{\mathcal H}(g_i, g_j)$ where $ d_{\mathcal H}$ denotes the Hellinger distance.
However, this choice of loss function leads to a more complex optimization problem, that cannot be solved with our approach.
Indeed, as discussed later in Section \ref{sec:ralm}, the positivity of the density $g_i$ might not be preserved by the intermediate steps of the algorithm.
Therefore, we need a loss function that continues to make sense for negative densities.

\subsection{The Optimization Space}

Consider now the space over which one should minimize \eqref{eq:postproc_obj}.
First of all, we must require the existence of $Q^{-1}$ to interpet $\Lambda^\prime = \Lambda Q^{-1}$.
Moreover, for the model to make sense we need to ensure the positivity of the coefficients involved, i.e. $\Lambda^\prime = \Lambda Q^{-1} \geq 0$ and $M^\prime = Q M \geq 0$.
Finally, we observe that (i) given an ``optimal'' $Q$ such that $L(Q; M, J, \theta) = 0$, $L(\gamma Q; M, J, \theta) = 0$ for any $\gamma >0$, and (ii) $L(Q; M, J, \theta)$ attains lower values when the entries in $Q$ are small. 
Despite the preference for small $Q$ in the optimization problem, the resulting model is invariant to such rescalings since it involves the normalization of the underlying random measures. 
Hence, to overcome both issues we propose to add a further constraint in the optimization problem, namely $\det Q = 1$, which prevents having several optimal solutions differing by a constant and does not allow for matrices with entries too close to 0.

In conclusion, we propose to optimize \eqref{eq:postproc_obj} over the special linear group $SL(H) = \{Q \in \mathbb R^{H \times H}: \det Q = 1\}$, with the additional positivity constraints, i.e. our optimization problem becomes
\begin{equation}\label{eq:postprocess_opt}
\begin{aligned}
	\min_{Q \in SL(H)} \sum_{h, k = 1}^H  L(Q; M, J, \theta) 
	\text{\ \ s.t. \ \ } &\Lambda Q ^{-1}\geq 0, \ Q M \geq 0.
\end{aligned}
\end{equation}
The special linear group is not a linear space, therefore common gradient-based optimization techniques cannot be used to solve \eqref{eq:postprocess_opt}.
However, we can take advantage of the differential structure of $SL(H)$. In fact, it is a Lie group (hence, a smooth differentiable Manifold)
with associated Lie algebra $\mathfrak {sl}(H) = \{ A \in \mathbb R^{H \times H}: \mbox{tr} A = 0 \}$.  See Appendix \eqref{sec:app_lie} for some basic details regarding Riemannian manifolds and Lie groups.

\subsection{A Riemannian Augmented Lagrangian Method}\label{sec:ralm}

We are now in place to state the algorithm. For notational convenience,
define the functions $c^1_{jh}(Q) = - (\Lambda Q ^{-1})_{jh}$ and $c^2_{hk} = -(Q M)_{hk}$. Denote with $c_j$ the collection of all such functions.  The positivity constraints are equivalent to $c_j \leq 0$ for all $j$'s.
Following the augmented Lagrangian method \citep{birgin2014alm}, we can deal with the constraints $\Lambda Q ^{-1}\geq 0$ and  $Q M \geq 0$ by introducing auxiliary parameters $\rho$, $\gamma_j$ and define the augmented loss function
\begin{equation}\label{eq:pen_loss}
	\mathcal L_\rho(Q, \gamma) = L(Q; M, J, \theta) + \frac{\rho}{2} \sum_j \max\left\{0, \frac{\gamma_j}{\rho} c_j(Q) \right\}
\end{equation}
Then, we can solve \eqref{eq:postprocess_opt} by alternating between minimizing \eqref{eq:pen_loss} for fixed values of $\rho$, $\gamma_j$ and updating $\rho$, $\gamma_j$ as in Algorithm \ref{algo:alm}. 
As in the usual augmented Lagrangian method, the constraints might be violated in the intermediate steps. Intuitively, the fact that the penalty term $\gamma_j$ is increased at every iteration if the constraint is violated should force the solution of the problem inside the feasible region.
See \cite{birgin2014alm} for convergence results of the augmented Lagrangian method.

\SetNlSty{textbf}{[}{]}
\begin{algorithm}[t]
	\textbf{input}{ Starting point $Q$, initial values $\rho$, $\gamma_{j}$, target threshold $\varepsilon^*$, initial threshold $\varepsilon$.}
	
	\DontPrintSemicolon
	\Repeat{$\varepsilon \leq \varepsilon^*$; $\|Q - Q^\prime\| \leq \varepsilon$ }{
		$Q = Q^\prime$ \\
		solve $Q^\prime = \arg\min_Q \mathcal L_\rho(Q, \gamma)$ for fixed $\rho, \gamma$ with theshold $\varepsilon$  using Algorithm \ref{algo:lie_rattle} \\
		$\gamma_j = \gamma_j + \rho c_j(Q^\prime)$ \\
		$\rho = 0.9 \rho$ $\varepsilon = \max\{\varepsilon^*, 0.9 \varepsilon\}$
	}
	\textbf{end}
	\caption{\label{algo:alm}Augmented Lagrangian Multiplier Method}
\end{algorithm}

\SetNlSty{textbf}{[}{]}
\begin{algorithm}[t]
	\textbf{input}{ Starting point $Q, P$, momentum $\tau$, stepsize $s$, threshold $\varepsilon$.}
	
	\DontPrintSemicolon
	\Repeat{$\|Q - Q^\prime\| \leq \varepsilon$ }{
		$P = \tau \left(P - s \Pi_{\mathfrak {sl}(H)}(\partial_Q \mathcal L_\rho(Q, \gamma), Q)\right)$ \\
		$Q = Q \exp_m(\chi P )$, $\chi = \cosh(- \log \tau)$ \\
		$P = \tau \left(P - s \Pi_{\mathfrak {sl}(H)}(\partial_Q \mathcal L_\rho(Q, \gamma), Q)\right)$ \\
	}
	\textbf{end}
	\caption{\label{algo:lie_rattle}Lie RATTLE Optimization}
\end{algorithm}

It is now left to discuss how to solve \eqref{eq:pen_loss} for fixed $\rho$ and $\gamma_j$. We propose to tackle this problem with the Riemannian dissipative RATTLE algorithm in \cite{franca_rattle}, reported for the special case of optimization over $SL(H)$ in Algorithm \ref{algo:lie_rattle}.
In particular, $\Pi_{\mathfrak{sl}(H)}$ is the projection over the Lie algebra $\mathfrak {sl}(H)$ while $\exp_m$ denotes the matrix exponential, which is a map $\mathfrak{sl}(H) \rightarrow SL(H)$.
Informally, Algorithm \ref{algo:lie_rattle} resembles an accelerated gradient method, where a momentum term is introduced to speed up the convergence.
We further have
\[
	\partial_Q \mathcal L_\rho(Q, \gamma)_{ij} = \frac{ \partial_Q \mathcal L_\rho(Q, \gamma)}{\partial Q_{ji}}
\]
(note the index flip $ij \rightarrow ji$, in other words $\partial_Q f(Q) = \nabla_Q f(Q)^\top$ where $\nabla$ stands for the usual Euclidean gradient).
Moreover, the following proposition gives a computationally convenient way of evaluating $\Pi_{\mathfrak{sl}(H)}$.
\begin{proposition}\label{prop:lie_proj}
Let $X$ an $H \times H$ real valued matrix. Then
\begin{equation*}\label{eq:sl_proj}
	\Pi_{\mathfrak{sl}(H)}(X) = (X - \mbox{diag}(X))^T  + \sum_{\ell = 1}^{H-1} X^*_\ell
\end{equation*}
where $\mbox{diag}(X)$ is the diagonal matrix with entries equal to the diagonal of $X$ and $X^*_\ell$ is a diagonal matrix whose only nonzero entries are the $(\ell, \ell)$-th and the $(\ell +1, \ell+1)$-th ones, which equal to $X_{i,i} - X_{i+1,i+1}$ and $-X_{i,i} - X_{i+1,i+1}$ respectively.
\end{proposition}

The parameters involved in the optimization problem are: the stepsize $s$ and momentum factor $\tau$ in Algorithm \ref{algo:lie_rattle} as well as the initial values $\rho$, $\gamma_j$ and the target and thresholds $\varepsilon^*$, $\varepsilon$ in Algorithm \ref{algo:alm}.
We suggest as defaults $s=10^{-6}$, $\tau=0.9$, $\rho=\gamma_j=10$, $\varepsilon^* = 10^{-6}$, $\varepsilon = 10^{-2}$. 
Finally, to set the starting point $Q$ we we solve the unconstrained optimization problem (equivalent to setting $\gamma_j = 0$ in \eqref{eq:pen_loss}) using Algorithm \ref{algo:lie_rattle} and use that solution as starting point for the constrained optimization.
The initial momentum term $P$ in Algorithm \ref{algo:lie_rattle} is always the zero matrix.

\subsection{The Label-Switching Problem}\label{sec:label_switch}

Observe that another source of non-identifiability comes from the labeling of $\mu^*_1, \ldots, \mu^*_H$.
Namely, the likelihood and the loss function \eqref{eq:postproc_obj} are invariant under permutation of the indices}$\{1, \ldots, H\}$, provided that the columns of $\Lambda$ are permuted as well.
This prevents the possibility of computing reliable posterior summaries of the $\mu^*_h$'s and $\Lambda$ from the MCMC chains.

We propose to post-process the output of our sampling algorithm to get rid of this problem. In particular, as in \cite{dunson_post_proc}, we propose to align the latent measures at each iteration to a given template.
Let $\hat \mu_1, \ldots, \hat \mu_H$ denote the template. For instance,
\[
	\hat \mu_h = \sum_{k=1}^K (Q^{(\ell)}M^{(\ell)})_{jk} J^{(\ell)}_k \delta_{\theta^{(\ell)}_k} 
\]
where we denote with the superscript $\ell$ the index of the MCMC sample. We choose $\ell$ to approximate the maximum a posteriori. $Q^{(\ell)}$ denotes the associated optimal transformation matrix obtained as outlined above.
Let $d(\hat \mu_h, \mu^\prime_j)$ denote a dissimilarity between two measures. Two specific choices are discussed later.
We align each $(\mu^{\prime(j)}_1, \ldots, \mu^{\prime(j)}_H) := Q^{(j)}(\mu^{*(j)}_1, \ldots, \mu^{*(j)}_H)$ to $\hat \mu_1, \ldots, \hat \mu_H$ by learning an optimal permutation $\sigma$ of $\{1, \ldots, H\}$, associated to a permutation matrix $P_\sigma$ that minimizes $\sum_h d(\hat \mu_h, \mu^{(j)\prime}_{\sigma(h)})$ by solving
\[
	\inf_{P \in \text{Perm}_H} \sum_{h, k=1}^H d(\hat \mu_h, \mu^{(j)\prime}_k) P_{hk}
\]
where $\text{Perm}_H$ denotes the space of $H \times H$ permutation matrices. Naively, this would require $H!$ computations.
Instead, we solve the relaxed optimization problem by looking for the $P$ stochastic matrix (i.e., rows and columns sum to one) that minimizes the objective above. That is, we solve for the Wasserstein distance between the empirical measures $\nu_1$ and $\nu_2$ defined as
\[
	\nu_1 = \frac{1}{H} \sum_{h=1}^H \delta_{\hat \mu_h}, \qquad  \nu_2 = \frac{1}{H} \sum_{k=1}^H \delta_{\mu^{(j)\prime}_k}
\]
where $\nu_i$ is a probability measure on the space of positive measures over $\Theta$.
Birkhoff's theorem ensures that the solution to the relaxed optimization problem is a permutation matrix.

As far as the dissimilarity $d(\hat \mu_h, \mu^\prime_j)$ is concerned,  in our examples we considered
\[
	d(\hat \mu_h, \mu^\prime_j) = \Big \|  \hat \mu_h(\Theta)^{-1} \int_\Theta f(y \mid \theta) \hat \mu_h(\dd \theta)  - \mu^\prime_j(\Theta)^{-1} \int_\Theta f(y \mid \theta) \mu^\prime_j(\dd \theta) \Big\|
\]
where $\|\cdot \|$ stands for the $L_2$ norm. This distance requires the numerical evaluation of a mixture density on a fixed grid, to compute the associated $L_2$ distance. This is easy when the dimension of the data space is small, typically when data are uni or bi-dimensional. See Appendix \ref{sec:align_highdim} for a more efficient alternative in higher dimensions.

\section{Simulation Study}\label{sec:simu}

We present two simulations to assess the performance of our model.
In all the examples, we consider Gaussian mixture models, i.e., $\theta^*_h = (\mu_h, \sigma^2_h)$ and $f(\cdot \mid \theta) = \mathcal{N}(\cdot \mid \mu, \sigma^2)$.
The scores $m_{hk}$ in the CoRM are gamma distributed and each $\mu^*_h$ is marginally a gamma process (before the truncation) with total mass equal to $1$ and base measure equal to the Normal-inverse-Gamma distribution, i.e. $G_0(\mu, \sigma^2) = \calN(\mu \mid \mu_0, \sigma^2/\lambda) IG(\sigma^2 \mid a ,b)$. We set $\mu_0$ equal to the empirical mean of the observations, $\lambda = 0.01$, $a=b=2$.
We truncate the CoRM to $K=20$ jumps to perform posterior inference. Specific choices of the prior for $\Lambda$ are discussed case-by-case.

\subsection{Interpretation of the posterior distribution}

Before giving details on the numerical illustration, we discuss how to obtain interpretable summaries of the posterior distribution, after post-processing. This also allows us to set some notation used in the next sections.

Interpreting the unnormalized \emph{latent factor densities}  $\int_\Theta f(\cdot \mid \theta) \mu^*_h(\dd \theta)$ is difficult because of the lack of a common scale to which the densities should be referred. In fact, note that these are not probability densities.
Let $p_j$ be the $j$-th group-specific density. We can write
\[
	p_j = \int_\Theta f(\cdot \mid \theta) \bar{\ptilde}(\dd \theta) + \sum_{h=1}^H s_{jh} \int_\Theta  f(\cdot \mid \theta) \epsilon_h(\dd \theta)
\]
where $\bar{\ptilde}(\dd \theta)$ is the average of $\ptilde_1, \dots, \ptilde_g$, $p^\prime_h = \mu^\prime_h / \mu^\prime_h(\Theta)$, $\epsilon_h = p^\prime_h - \bar{\ptilde}(\dd \theta)$ and the scores $s_{jh}$'s are defined as
\begin{equation}\label{eq:convex_scores}
	s_{jh} = \frac{\lambda^\prime_{jh} \mu^\prime_h(\Theta)}{\sum_{k=1}^H \lambda^\prime_{jk} \mu^\prime_k(\Theta)} 
\end{equation}
Note that $\epsilon_h$ is a signed measure.
Instead of comparing the latent factor densities, we find it considering the {\it residual factor densities} $\int_{\Theta} f(\cdot \mid \theta) \epsilon_h(\dd \theta)$ leads to easier interpretations.

Moreover, we can associated to each $\mu^\prime_h$ an \emph{importance score} $I_h$ defined as $I_h = \sum_{j=1}^g  s_{jh}$
The rationale comes from writing $\mu^\prime_h =  \mu^\prime_h(\Theta) p^\prime_h$ so that
\[
	p_j = \int_\Theta f(\cdot \mid \theta) \sum_{h=1}^H \frac{\lambda^\prime_{jh} \mu^\prime_h(\Theta)}{ \sum_{k=1}^H \lambda_{jk} \mu^\prime_k(\Theta)} p^\prime_h(\dd \theta) = \sum_{h=1}^H s_{jh} \int_\Theta f(\cdot \mid \theta) p^\prime_h(\dd \theta) 
\]
that is, we express each $\widetilde p_j$ as a convex combination of probability measures and with weight $s_{jh}$.

With an abuse of notation, we will denote by $\mu^\prime_h$ the posterior mean of $(\mu^*_1, \ldots, \mu^*_h)$ and with $\Lambda^\prime$ the posterior mean of $\Lambda$, obtained after the post-processing of the MCMC chains, that is
\begin{equation}
\begin{aligned}
	\mu^\prime_h = \frac{1}{M} \sum_{\ell=1}^M \sum_{k \geq 1} \left(P^{(\ell)} Q^{(\ell)} M^{(\ell)} \right)_{h k} J^{(\ell)}_k \delta_{\theta^{*(\ell)}_k} , \quad 
	\Lambda^\prime = \frac{1}{M} \sum_{\ell=1}^M \left(\Lambda^{(\ell)} (Q^{(\ell) })^{-1}\right) (P^{(\ell)})^\top
\end{aligned}
\end{equation}
where the superscript $\ell$, $\ell = 1, \ldots, M$ is used to denote the iteration of the MCMC algorithm, $Q^{(\ell)}$ is the matrix found with Algorithm \ref{algo:alm}, and $P^{(\ell)}$ is the permutation matrix found as in Section \ref{sec:label_switch}.

\subsection{Only Group Information}\label{sec:simu1}

We consider here a simulated example with $g=100$ groups of data, where each $n_j = 25$.
We consider the situation where we tend to observe only small differences across populations by considering the following data generation process
\[
	y_{j, i} \iid  w_{j1} \,\calN(-2, 2) +  w_{j2}\, \calN(0, 2) +  w_{j1} \,\calN(2, 2), \qquad i=1, \ldots, n_j
\]
and for each group we simulate $\bm w_j = (w_{j1}, w_{j2}, w_{j3}) \iid \mbox{Dirichlet}(1, 1, 1)$. In most of the groups, the data generating density is unimodal and they differ mainly because of different levels of skewness.

As prior for $\Lambda$, we assume the multiplicative gamma process \eqref{eq:lambda_mgp} setting $H=20$. We run the MCMC chains for a total of $11,000$ of which the first $1,000$ are used for the adaptation and the following $5,000$ are discarded as burn-in.
The adaptation phase quickly finds between 3 and 5 latent measures, 4 being the final value.
We post-process the chains as in Section \ref{sec:postproc}.

Figure \ref{fig:mgp_latent} shows the inferred latent factors densities before and after the post-processing. It is clear that solving the label switching is essential.
Although not particularly evident from the plot, the matrices $Q^{(j)}$ found by the optimization algorithm were significantly different from the identity, hence showing the usefulness of the post-processing.
Our approach identifies the main common traits in the data. Factors 1 and 3  peak around $-2$ and $2$ respectively, while the second and fourth factors are both more concentrated around the origin, with the second one presenting a light skewness and heavier right tail.
The residual factor densities can be used to infer the same description of the latent measures.

\begin{figure}[t]
	\centering
	\includegraphics[width=\linewidth]{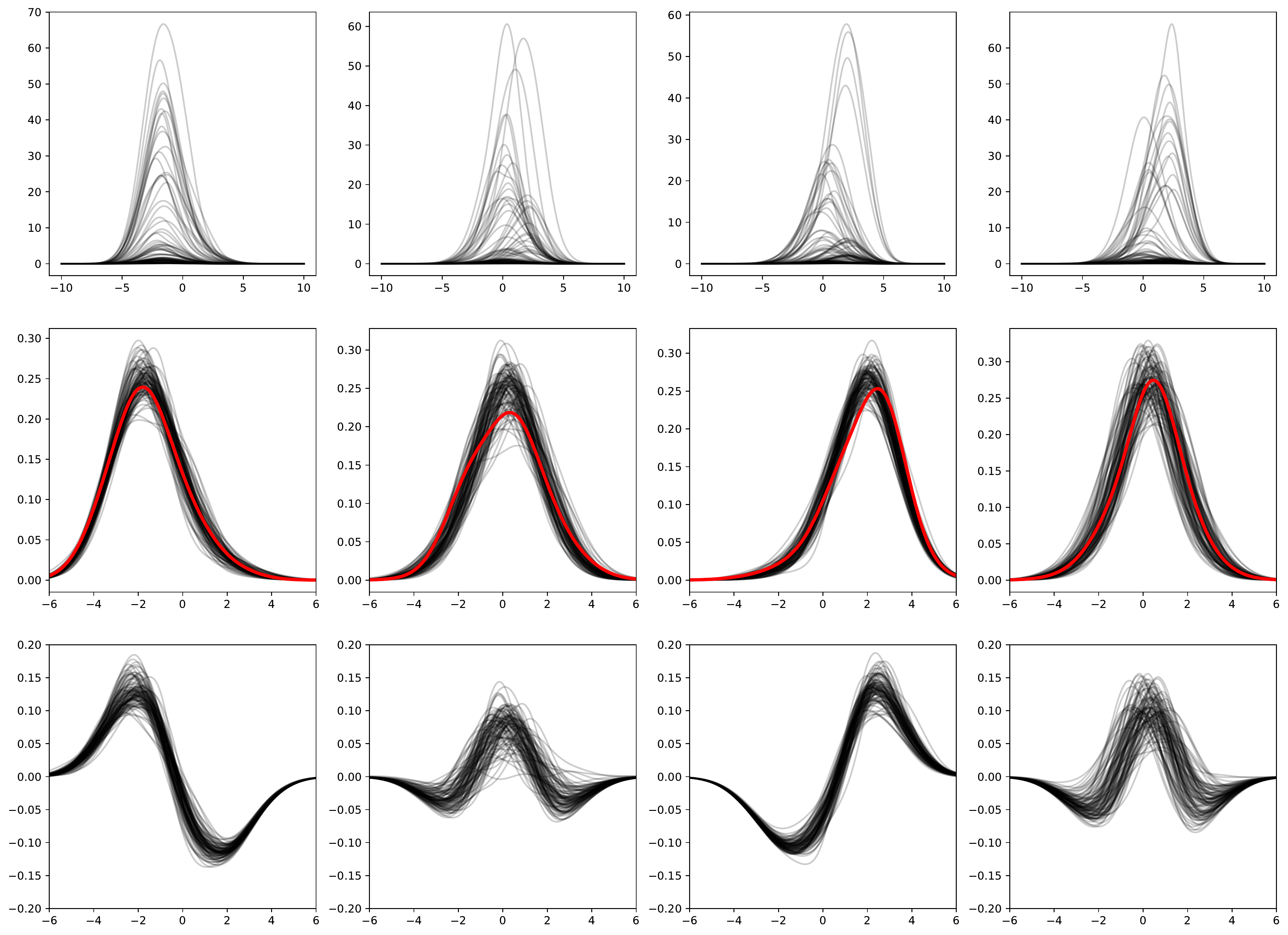}
	\caption{Posterior summaries for the simulation in Section \ref{sec:simu1}.
	Top row: draws from the posterior distribution of the latent factor densities. Middle row: draws after post-processing and normalization, the red density denotes the template. Bottom row: posterior draws of the residual factor densities.}
	\label{fig:mgp_latent}
\end{figure}

\subsection{Area-Referenced Data}

We consider data over a regular lattice on ${0, 1, \ldots, q} \times  {0, 1, \ldots, q} \subset \mathbb{Z}^2$. We consider $q = 4, 8, 16$ so that the number of groups is $g = 16, 64, 256$ respectively.
Following the simulation study in \cite{beraha_spmix}, we generate data at each location from a three-component Gaussian mixture with means $-5, 0, 5$ respectively and variances equal to one.
Let $x_j, y_j$ denote the $x$ and $y$ coordinate of location $j$ on the lattice. The location-specific weights are
\[
    (w_{j1}, w_{j2}, w_{j3}) = \left(e^{\widetilde{w}_{j1}}, e^{\widetilde{w}_{j2}}, 1 \right) \big/ \left( 1 + e^{\widetilde{w}_{j1}} + e^{\widetilde{w}_{j2}} \right)
\]
where 
\[
    \widetilde{w}_{j1} =  3(x_j - \bar x) + 3(y_j - \bar y), \quad \widetilde{w}_{j2} =  -3(x_j - \bar x) - 3(y_j - \bar y)
\]
and $(\bar x, \bar y)$ denote the center of the lattice. For each location, 25 observations are simulated.

We compare our model with prior \eqref{eq:lambda_gmrf} for $H=1, 2, 3, 5, 10$ with the spatially dependent mixture model \cite[SPMIX,][]{beraha_spmix} and the Hierarchical Dirichlet Process \citep[HDP,][]{teh2006hierarchical}. 
Although the latter does not take into account the spatial dependence, it is shown in \cite{beraha_spmix} that the HDP performs well when the number of groups $g$ is small.

We truncate the CoRM to $K = 20$ jumps and set the number of components in SPMIX to 20 as well. 
Prior distributions can be assumed for $\tau$ and $\rho$ in \eqref{eq:lambda_gmrf}. However, since the likelihood is invariant with respect to rescalings of $\Lambda$, we found that having a prior on $\tau$ led to non-convergent MCMC chains for $\Lambda$. 
In particular, after a few thousand iterations, the values of the entries in $\Lambda$ were in the order of $10^{100}$.
Hence, we suggest fixing $\tau$ so a sufficiently large value. In our simulations, we always set $\tau \equiv 2.5$.
Assuming a prior for $\rho$ does not have such an impact on posterior inference. However, it would require re-computing the determinant of $\Sigma^{-1}$ at every MCMC iteration, which requires $O(g^3)$ operations.
Hence, we fix $\rho$ to $0.95$ to encourage strong spatial dependence in our examples.
Another possibility would be to fix a grid of values in $(0, 1)$ and assume a discrete prior for rho over it, allowing to compute all the required matrix determinants beforehand. 

All the MCMC chains are run for $10,000$ iterations, discarding the first $5,000$ as burn-in.
It is clear from Figure \ref{fig:spatial_simu} (top row) that our model outperforms the competitors when $g = 16, 64$ and performs slightly better than the spatial mixture model when $g=256$.
In all the settings, the best performance is associated with $H=3$ latent measures.
Posterior samples of the latent factor densities are reported in Figure \ref{fig:spatial_simu} (bottom row) for the setting with $g=64$ and $H=3$. In this case, the latent densities are already well separated so that there is no need to post-process the MCMC chains using the algorithm described in Section \ref{sec:postproc}.
The three latent densities give mass to one of the three modes in the data each.

\begin{figure}
    \centering
    \includegraphics[width=\linewidth]{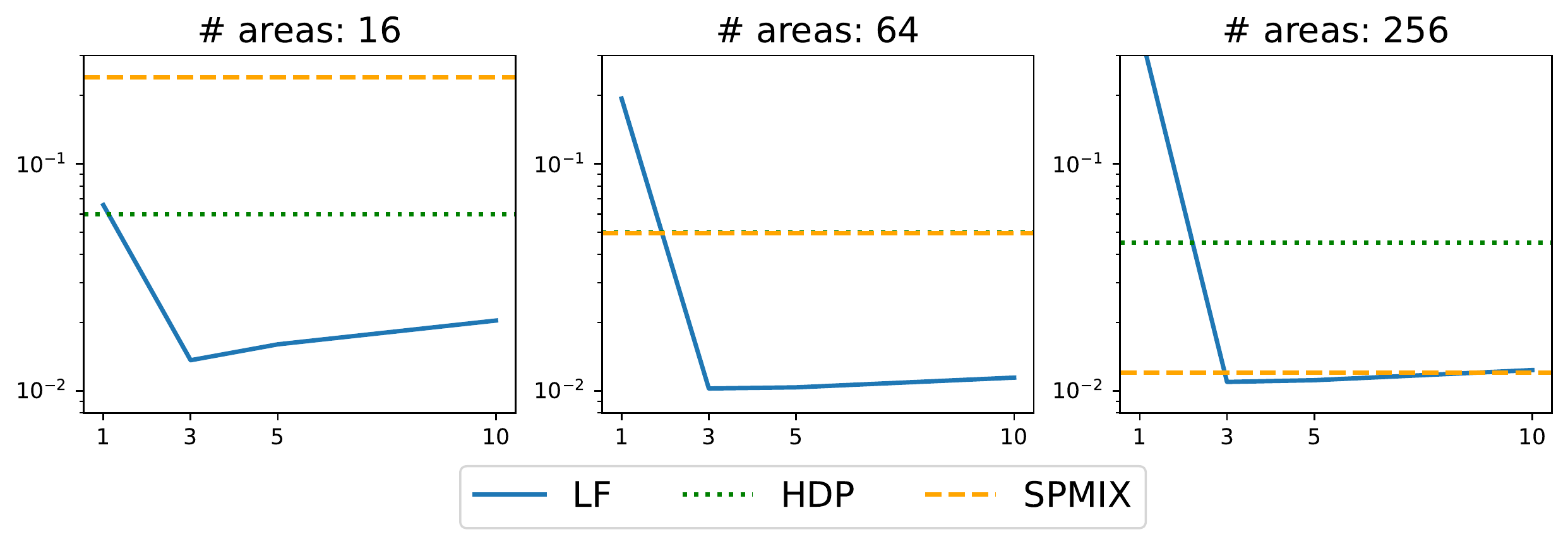}
    \includegraphics[width=\linewidth]{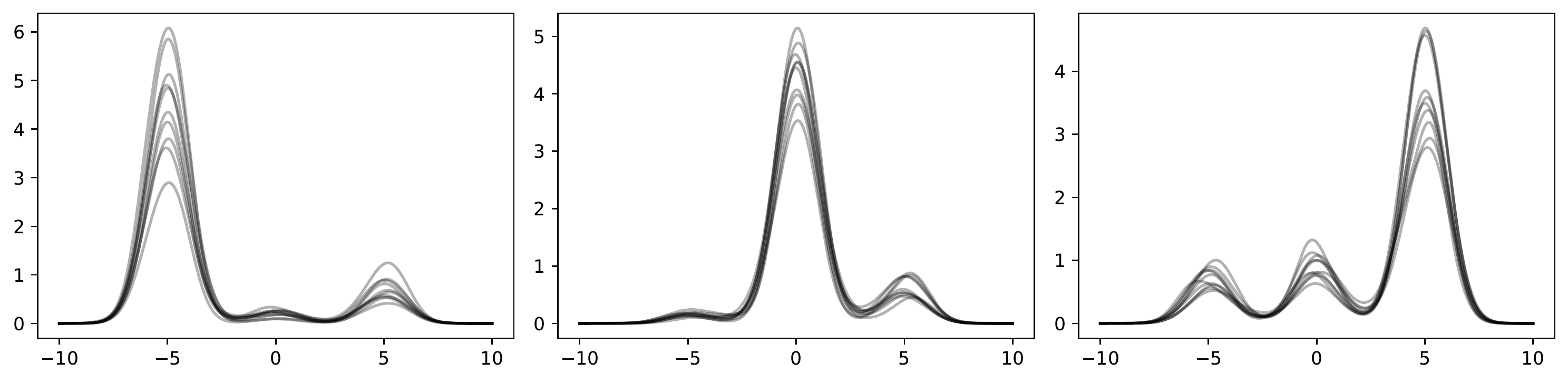}
    \caption{Top row: Average Kullback--Leibler divergence between the true data generating density and the Bayesian estimate, as a function of the number of latent measures $H$. From left to right $g=16, 64, 256$.
    Bottom row: Posterior samples for the latent factor densities when $g=64$ and $H=3$}
    \label{fig:spatial_simu}
\end{figure}

\section{Real Data Illustrations}\label{sec:real_data}

In this section, we illustrate our methodology on two real datasets. In both cases, data are univariate and we let $f(\cdot \mid \theta)$ be the Gaussian density with parameters $\theta = (\mu, \sigma^2)$. The base measure $G_0$ is the Normal-inverse-Gamma distribution, whose parameters are set as in Section \ref{sec:simu}.
Moreover, we always truncate to $K=20$ points the support of the random measures.

\subsection{The Invalsi Dataset}\label{sec:invalsi}

\begin{figure}[t]
\centering
\includegraphics[width=\linewidth]{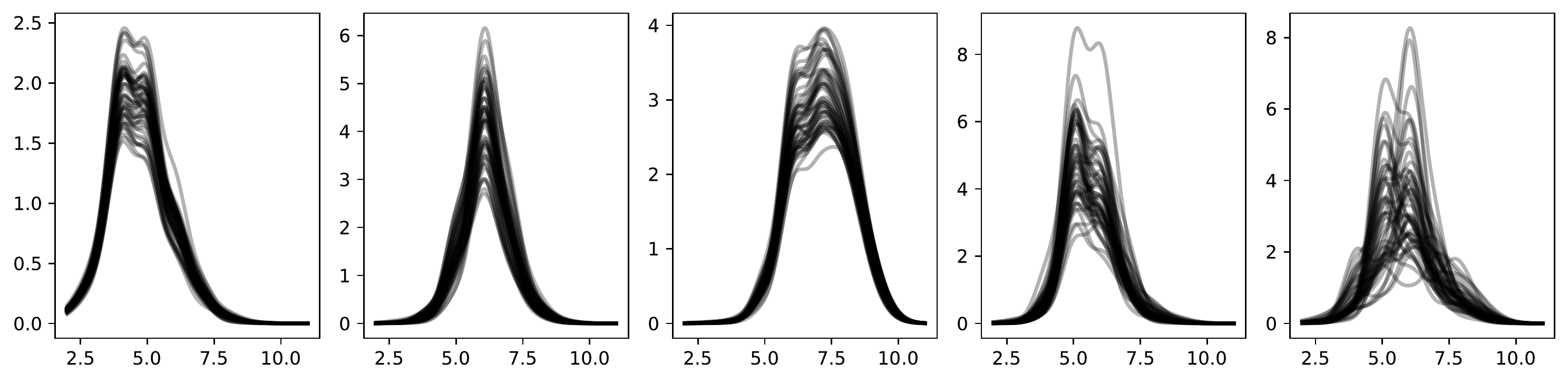}
\includegraphics[width=\linewidth]{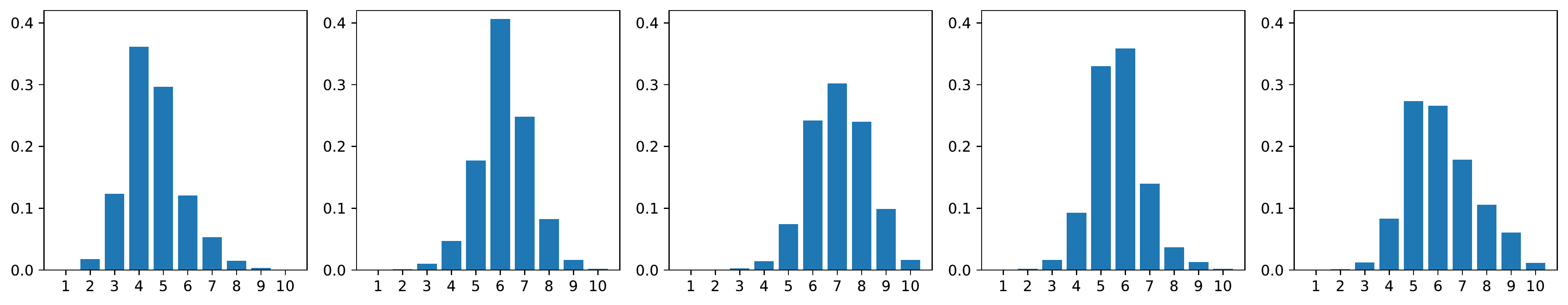}
\includegraphics[width=\linewidth]{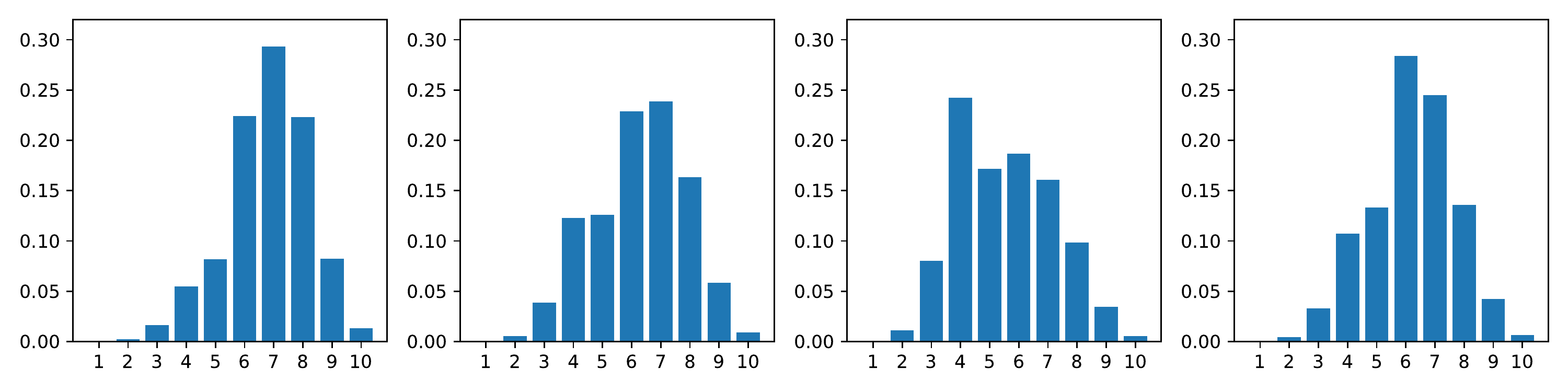}
\caption{Summary of posterior inference on the Invalsi dataset. Top row: draws from the posterior distribution of the latent factor densities. 
Middle row: estimates of the discretized normalized latent factor densities after post-processing. 
Bottom row: average density in each cluster discredized on the intervals $[i - 0.5, i + 0.5)$, $i=1,\ldots,10$.}
\label{fig:invalsi_latent}
\end{figure}

We consider the \textit{Invalsi} dataset\footnote{available for research purposes at \url{https://invalsi-serviziostatistico.cineca.it}} that collects the evaluation of a unified math test undertaken by all Italian high-school students. 
Grades vary from $1$ to $10$ with 6 being the passing grade. We pre-process the data by adding a small Gaussian noise with zero mean and standard deviation equal to 0.25.
The dataset contains the scores of 39377 students, subdivided into 1048 schools. The number of students per school varies from 4 to 131, with 37 students per school on average with a standard deviation of 12 approximately.

We assume the multiplicative gamma process prior for $\Lambda$ as in \eqref{eq:lambda_gmrf} with $H=20$. The initial adaptation phase identifies 5 latent factors.
Draws from the latent factor densities are displayed in Figure \ref{fig:invalsi_latent}. It is clear that some label switching is happening between the fourth and fifth factors.
After the post-processing, for ease of visualization, we discretized the estimated normalized latent factor densities to the original grades $i = 1, \ldots, 10$ by evaluating $\int_{i-0.5}^{i+0.5} f(y \mid \theta) \mu^\prime_h(\dd \theta) / \mu^\prime_h(\Theta)$.
The estimated factors are displayed in the first two rows of Figure \ref{fig:invalsi_latent}.
They represent a wide range of behaviors: the first one is concentrated on negative grades below the passing threshold, the second one is centered on the passing grade, and the third one on grades way above the passing grade.
The fourth and the fifth represent more complex distributions: the former one covering the range of ``just below the passing grade and just above it'', the latter one instead represents a distribution peaked at 5 with a heavy right tail.

The importance scores $I_h$ are approximately $331, 184, 351, 165, 16$. Hence, we can interpret that the two most relevant common traits are the ones represented by $\mu^\prime_1$ (that combines a sharp peak in 4, with a heavy right tail), and by $\mu^\prime_3$, which gives mass to grades above the passing threshold.

Finally, we look at the scores $\lambda_{jh}$'s after the post-processing. We can understand the similarities between schools by clustering the scores for each school from
the corresponding row of the matrix $\Lambda^\prime$. Using  a hierarchical clustering algorithm
 yields four clusters (the dendrogram is shown in Figure~\ref{fig:invalsi_dendrogram} in the Appendix). We then compute the average value $\hat \lambda_\ell = (\hat \lambda_{\ell 1}, \ldots, \hat \lambda_{\ell H})$ for each of the four clusters, to which a probability measure $\ptilde_\ell \propto \sum_{h=1}^H \hat \lambda_{\ell h} \mu^\prime_h$ and report the associated mixture density in the bottom row Figure~\ref{fig:invalsi_latent}.We define a cluster-specific mean distributions $\ptilde_\ell \propto \sum_{h=1}^H \hat \lambda_{\ell h} \mu^\prime_h$ by taking the average value $\hat \lambda_\ell = (\hat \lambda_{\ell 1}, \ldots, \hat \lambda_{\ell H})$ for each of the four cluster.
the associated mixture densities are shown in the bottom row Figure~\ref{fig:invalsi_latent}.
The clusters are easily interpretable and the mean distributions $\ptilde_1, \dots \ptilde_4$ are substantially different.

\subsection{Californian Income Data}\label{sec:income_data}

\begin{figure}[t]
\centering
\includegraphics[width=\linewidth]{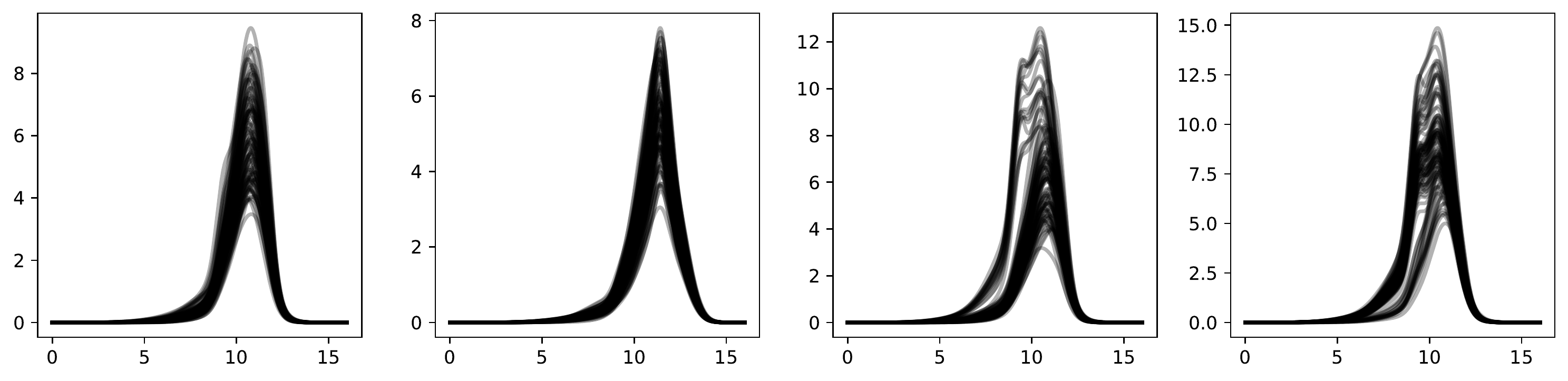}
\includegraphics[width=\linewidth]{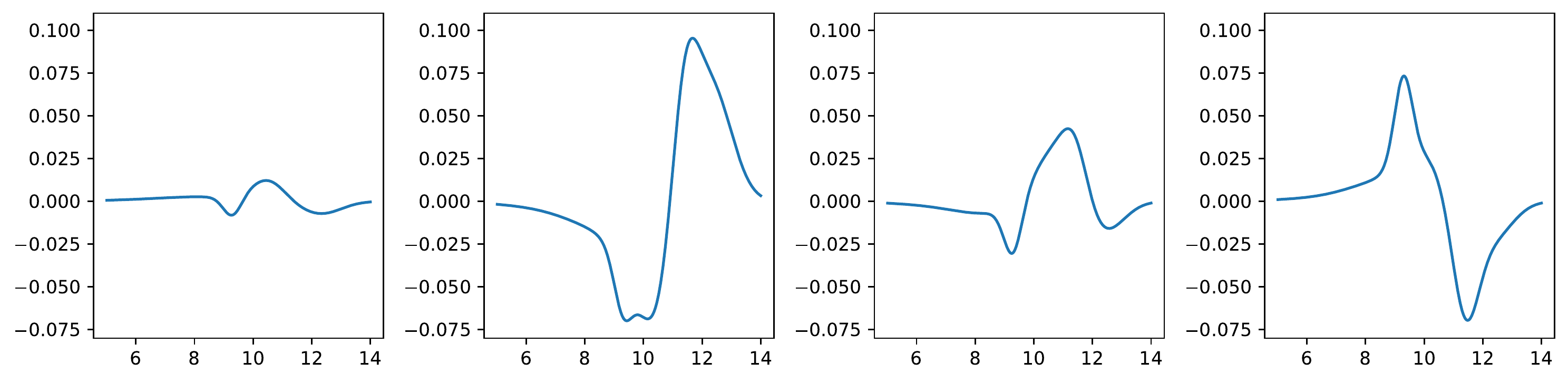}
\caption{Summary of posterior inference on the Californian income dataset.
Top row: draws from the posterior distribution of the latent factor densities. 
Bottom row: average of the residual factor densities after post-processing.}
\label{fig:income_post}
\end{figure}

\begin{figure}[h!]
\centering
\includegraphics[width=\linewidth]{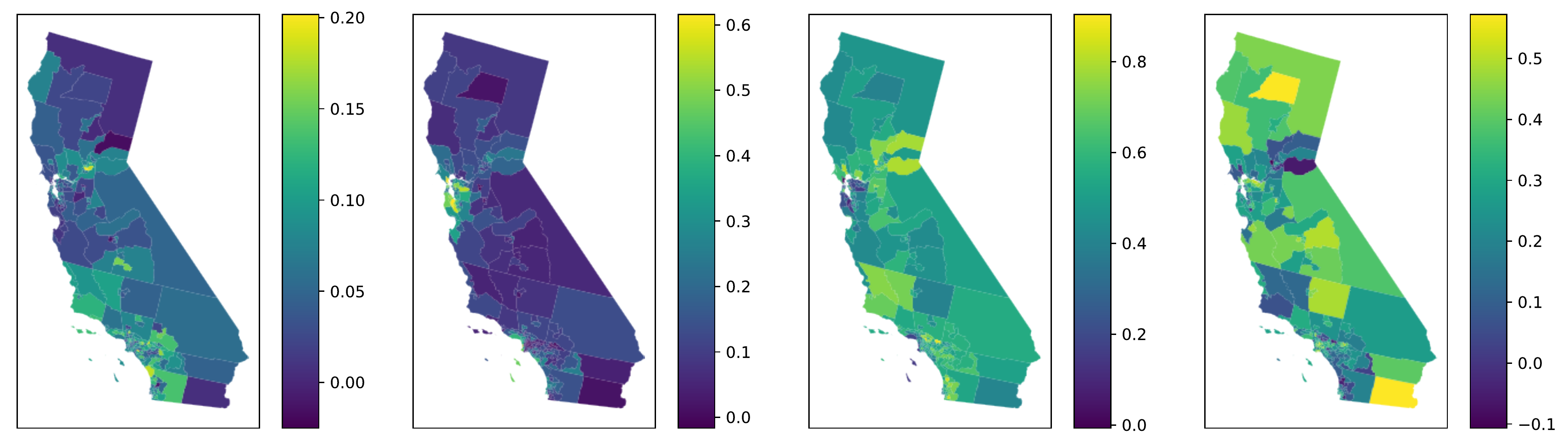}
\includegraphics[width=\linewidth]{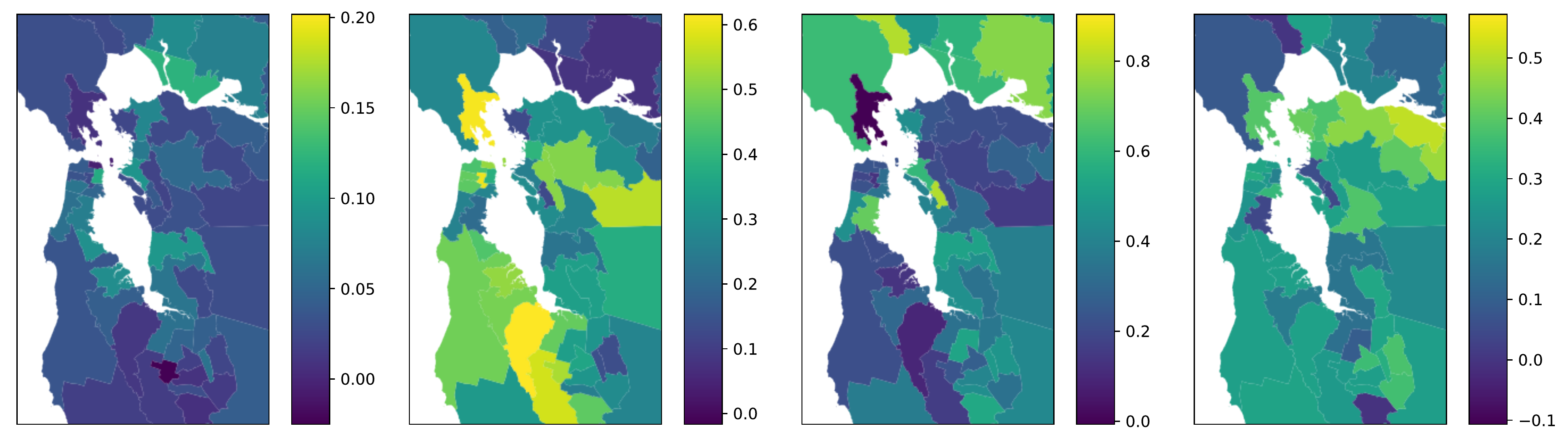}
\includegraphics[width=\linewidth]{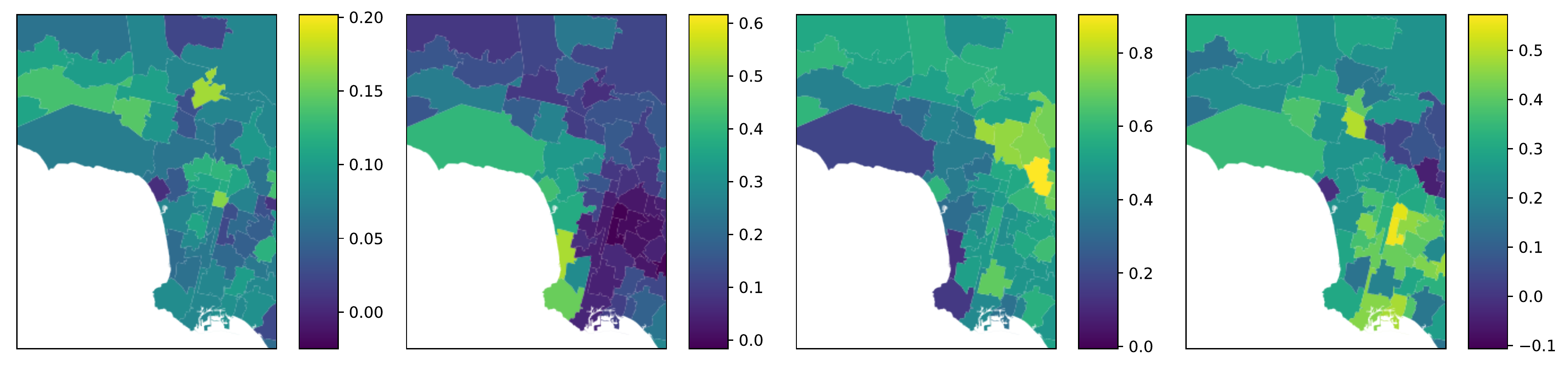}
\caption{Spatial distribution of the scores in the Californian income dataset. Top row: the scores $s_{jh}$ for $h=1, \ldots, 4$ from left to right. Middle row: zoom on the San Francisco area. Bottom row: zoom on the Los Angeles area}
\label{fig:income_map}
\end{figure}

We consider the 2021 ACS census data publicly available at \url{https://www.census.gov/programs-surveys/acs/data/experimental-data/2020-1-year-pums.html}.
Specifically, we consider the \texttt{PINCP} variable that represents the personal income of the survey responders and restrict to the citizens of the state of California.
For privacy reasons, data are grouped into geographical units denoted as \emph{PUMAs}, roughly corresponding to $100,000$ inhabitants. There are 265 PUMAs in California.
We consider $y_{j,i}$ to be the logarithm of the income of the $i$-th person in the $j$-th PUMA. The total number of responders is $43380$, with the median number of observations per PUMA being 164.

As shown in Figure \ref{fig:income_eda} in the Appendix, the distributions of the income in different PUMAs are quite varied with clear spatial dependence. This is also confirmed by the analysis of Moran's $I$ index for the average log-incomes, which is approximately $0.55$. A permutation test confirmed that the spatial correlation is not-negligible.
We assume independent log Gaussian Markov random fields priors for each column of $\Lambda$ as in \eqref{eq:lambda_gmrf}, where we fix $\tau = 2.5$ and $\rho = 0.95$. 
We choose $H$ by evaluating the predictive goodness of fit for $H=1, \ldots, 10$ using the widely applicable information criterion \citep[WAIC,][]{watanabe2013widely}. The best performance is associated with $H=4$, therefore we comment on the posterior inference obtained under this model.

Figures \ref{fig:income_post} and  \ref{fig:income_map} summarize the posterior findings. 
The draws from the latent factor densities (top row) show some evidence of label-switching in the third and fourth factors.
Post-processing the chains with our algorithm estimates the four latent factors in Figure \ref{fig:income_latent_est} in the Appendix. However, it is easier to interpret the residual factor densities displayed in the bottom row of Figure \ref{fig:income_post}.
The second and the fourth factors are associated with the largest variations. In particular, the second one gives mass to higher incomes while the fourth one gives mass to lower incomes. The first one is more representative of the average population since the variations are small.
The third factor instead corresponds to average incomes and gives less mass (compared to the average population) to both low and high incomes.
To visualize the spatial effect of the latent factors, we plot the scores $s_{jh}$ for each factor.
Note that the third latent factor is predominant in several areas, where $s_{j3}$ is larger than 0.8. Instead, $s_{j2}$ is small in all of California except for a few PUMAs in San Francisco, Long Beach, and San Diego, where the highest incomes are observed.
In particular, zooming on San Francisco (middle row of Figure \ref{fig:income_map}), we note that the second factor is highly represented in Palo Alto, home to several tech tycoons, and San Rafael, home to entertainers. 
Finally, note that the fourth factor (associated with the lowest incomes) has a high weight in the two PUMAs neighboring Mexico as well in some areas in Los Angeles. Notably, the PUMA around the port and the one corresponding to the ``south LA'' neighborhoods going from University Park to Green Meadows. This is in accordance with the 2008 \emph{Concentrated Poverty in Los Angeles} report \citep{poverty_la}, which estimates that  the percentage of households in poverty is typically above 40\% in those areas.

\section{Discussion}\label{sec:discussion}

Modeling a collection of random probability measures is an old problem that has received considerable attention in the Bayesian nonparametric literature, see, e.g. \cite{quintana2022dependent} for a recent review.
In this article, we have considered specifically the case when data are naturally divided into groups or subpopulations, and data are partially exchangeable.
Taking a nonparametric Bayesian approach, we assumed that observations in each group can be suitably modeled by a mixture density, and proposed
\emph{normalized latent measure factor models} as a prior for the collection of mixing measures in each group. 
Similar to the Gaussian latent factor model, our model assumes that each group-specific directing measure is a linear combination of a set of latent random measures.
We can interpret the latent random measures as the latent common traits shared by the subpopulations. 
Moreover, the prior for the linear combination weights can include additional group-specific information such as geographical location.

To account for the non-identifiability of our model, we developed an ad-hoc post-processing algorithm leading to a constrained optimization algorithm over the special linear group, that is the group of matrices whose determinant is equal to one. To solve the optimization problem, we leveraged recent work on optimization on manifolds, proposing a Riemannian augmented Lagrangian method.
Through simulations and illustrations on two real datasets, we validate our approach and show its usefulness, focusing in particular on the interpretation of the latent measures and the associated weights. The model opens up many directions for future research which we discuss below and aim to investigate thoroughly in the future.

Our factor model approach can be extended to a wide range of dependence structures between the groups. For example,  including observation-specific covariates in the model or time-dependent data. We can also build models which allow for the discovery of latent structure in the groups by further modeling the loadings matrix $\Lambda$. For instance,  \cite{rodriguez2008nested}, \cite{camerlenghi_nested}, and \cite{beraha2021semi} build models which cluster groups  according to the homogeneity of their distributions. We could achieve this by assuming that each of the group-specific directing measures is equal to one of the latent measures, {\it i.e.} only one of $\lambda_{j1},\dots, \lambda_{jH}$ are non-zero, which would be similar to exploratory factor analysis \citep{conti_BEFA}.  Alternatively, we can achieve a ``soft clustering'' of the group-specific distributions {\it i.e.} cluster together similar distributions as opposed to homogeneous ones by assuming a mixture model for the rows of the matrix $\Lambda$. More generally, $\Lambda$ could be expressed in terms of further low-rank matrix to find similarities between the group-specific factor loadings.

The post-processing identification scheme leads to estimated latent measures which are maximally separated according to the interpretability criterion. This allows us to interpret the factor loadings as an $H$-dimensional summary of the group-specific distribution where each element of the summary measures different parts of the distribution. In a similar way to scores from dimension reduction techniques, such as Principal Components Analysis, or embeddings in machine learning, these estimates can then be used as inputs into other statistical analyses. We effectively use this idea in the analysis of the Invalsi data-set where the estimated group-specific factor loadings are clustered to find groups of schools with similar distributions. This approach could have much wider applications. For example, the analysis of the Californian income data leads to estimated factor loadings for each PUMA which could be used in a regression model in place of other summaries such as median income, or the percentage of incomes below/above a threshold. These estimated factor loadings should provide more information than a single measure and be a more efficient representation than a large number of measures (for example, using a large number of thresholds). It would be particularly interesting to investigate this approach beyond univariate data, such as continuous or discrete multivariate observations where it's difficult to find efficient low-dimensional summaries of distributions.

\section*{}
\bibliographystyle{chicago}
\bibliography{ref_nrmi_factors}

\FloatBarrier

\appendix

\section{Technical Preliminaries}

\subsection{Completely Random Measures}

Let $\mathbb{M}_{\Theta}$ be the space of boundedly finite (positive) measures over the space $(\Theta, \mathcal{B}(\Theta))$, where $\mathcal{B}(\Theta)$ is the Borel $\sigma$-algebra.
We endow $\mathbb{M}_{\Theta}$ with the corresponding Borel $\sigma$-alebra $\mathcal M$.
Then, a random measure is a measurable function from a base probability space $(\Omega, \mathcal F, \mathbb P)$ to $(\mathbb{M}_{\Theta}, \mathcal M)$.

Following \cite{Kin67}, we say that a random measure $\mu$ is completely random if, for any $\{A_1, \ldots, A_m\} \subset \mathcal{B}(\Theta)$, $A_i \cap A_j = \emptyset$ ($i \neq j$), we have that the random variables $\mu(A_i)$, $i=1, \ldots, n$ are independent.

For our purposes, it is sufficient to consider completely random measures of the kind
\[
	\mu(A) = \int_{ \mathbb R_+ \times A} s N(\dd s \dd x)
\]
where $N$ is a Poisson point process on $\Theta \times \mathbb R_+$ with base (intensity) measure. We will assume that the intensity measure factorizes as $\nu(\dd s) G_0(\dd x)$ where $\nu$ is a Borel measure on the positive reals and $G_0$ is a probability measure on $\Theta$.
Then, the random measure $\mu(A)$ is uniquely characterized by its Laplace transform, for any measurable $f$, $f(x) \geq 0$:
\[
	\E \left [e^{-\int_{\Theta} f(x) \mu(\dd x)} \right] = \exp\left(- \int_{\R_+ \times \Theta} \left( 1 - e^{- s f(x)}\right) \nu(\dd s)  G_0 (\dd x) \right),
\]
where the equality follows from the L\'evy-Khintchine representation of the underlying Poisson process.

A key result that will be used later, is the Cambell-Little-Mecke formula (also referred to as the Palm formula) which allows the interchange of expectation and integral when the integrand measure is a point process.
We report here the result for Poisson point processes, the most general case can be found in \cite{BaBlaKa}.

\begin{theorem}{[\emph{Campbell-Little-Mecke}]}

Let $N$ be a Poisson point process over a complete and separable metric space $\X$ with intensity measure $\nu(\dd x)$. Denote by $\mathbb{M}_\X$ the space of boundedly $\sigma$-finite measures on $\X$. Then, for any measurable $g: \X \times \mathbb M_{\X} \rightarrow \R_+$ we have
\begin{equation}\label{eq:clm}
    \E\left[ \int g(x, N) N(\dd x) \right] = \int \E [g(x N + \delta_{x})] \nu(\dd x)
\end{equation}
where both expectations are with respect to the law of the Poisson process $N$.
\end{theorem}

\subsection{Riemannian Manifolds and Lie Groups}\label{sec:app_lie}

A group $G$ is a set equipped with a binary operation: $G \times G \rightarrow G$ with the additional properties that the operation is associative, there exists an identity element and every element has its inverse.
A Lie group arises if the set is a differentiable manifold and the binary and inverse operations are smooth differentiable functions.
A classic example of a Lie group is the set of $2 \times 2$ real-valued invertible matrix, endowed with the group operation $(A, B) \mapsto AB$, that is the standard matrix multiplication. This group is usually referred to as the \emph{general linear group} of dimension two and is denoted by $GL(2, \mathbb{R})$.

For our purposes, it is sufficient to consider matrix Lie groups, i.e., the case when $G$ is a set of matrices, so that $G \subset \mathbb R^{n \times n}$ for some $n$.
We can thus endow $G$ with the Riemannian metric induced by the Euclidean metric in $\mathbb R^{n^2}$
Then $G$ is a Riemannian manifold (it locally resembles a Euclidean space), and we can define at each point $g \in G$ a tangent space $T_g G$ together with the maps $\exp_g: T_g G \rightarrow G$ and $\log_g: G \rightarrow T_g G$.

The tangent spaces in Lie groups admit a particularly simple representation. Thanks to the fact that left multiplication by an element $g \in G$, that is the map $L_g(x) = gx$, is a diffeomorphism whose inverse is $(L_g)^{-1} = L_{g^{-1}}$, we have that the tangent space $T_g G$ at $g$ is isomorphic to $T_I G$, where $I$ is the identity element. The differential of $L_g$ is an isomorphism between $T_I$ and $T_g$.
In particular, given $v \in T_I  G$, we have that $g \exp (v) \in T_g G$. 
Therefore, it is sufficient to study only one tangent space, namely $T_I G$ that is the tangent space at the identity element.
This space is usually referred to as the Lie algebra, since it can be endowed with an additional operation (the Lie bracket) which makes it indeed an algebra.
When we consider Lie groups of matrices, the Lie algebra is again a set of matrices and the map $\exp(v)$ is simply the matrix exponential, i.e.
\[
	\exp(v) = \exp_m(v) = \sum_{n = 0}^\infty \frac{v^{n}}{n!}
\]
which is easily approximated by a variety of numerical algorithms.

\section{Proofs}\label{sec:proofs}

\subsection{Proof of Proposition \ref{prop:levy_crm}}
\begin{proof}
    Let $H=1$, then the L\'evy-Khintchine representation entails
    \begin{align*}
        & \E \left[ \exp\left( - \int_\Theta f(x) \mutilde_j(\dd x) \right) \right] =  \E \left[\exp\left( - \int_\Theta f(x) \lambda_{j1} \mu^*_h(\dd x) \right) \right] = \\
        & \qquad \exp \left( - \int_{\R^+ \times \Theta} \left( 1 - \exp(- s \lambda_{j1} f(x) )\right) \rho^*_h(s) \dd s \, \alpha^*_h(\dd x) \right) = \\
        & \qquad \exp \left( - \int_{\R^+ \times \Theta} \left( 1 - \exp(- s^\prime f(x) )\right) \rho^*_h(s^\prime/\lambda_{j1}) \lambda_{j1}^{-1} \dd s^\prime \, \alpha^*_h(\dd x) \right)
    \end{align*}   
    where the last equality follows from the change of variables $s^\prime = \lambda_{j1} s$. This proves the claim when $H=1$.
    
    In the more general case $H > 1$, we have that $\mutilde_j$ is the superposition of the random measures $\lambda_{j1} \mu^*_1, \ldots, \lambda_{jH} \mu^*_H$, which are independent since the $\mu^*_h$'s are. Hence, the L\'evy intensity of $\mutilde_j$ is the sum of the intensities of the $\lambda_{jh} \mu^*_h$'s.
    \end{proof}

\subsection{The latent factor model is not completely random}

From representation \eqref{eq:lat_corm} it is easy to see that $\mutilde_1, \ldots, \mutilde_g$ is not a vector of completely random measures.
Indeed, for any two disjoint measurable sets $A, B$ the random variables defined as
\[
	\mutilde_j(A) = \sum_{k \geq 1} \gamma_{jk} J_k \indicator[\theta^*_k \in A]
\]
are not independent.
This is due to the the scores $\gamma_{jk} = (\Lambda M)_{jk}, k=1, \ldots, $ which are not a collection of independent random variables.

\subsection{Proof of Theorem \ref{teo:expectation}}

We first state a technical lemma providing an alternative characterization of compound random measures.

\begin{lemma}
Let $\pi_h: \R^H \rightarrow \R$ be the canonical projection along the $h$-th coordinate, i.e. $\pi_h(\bm x) = x_h$ for all $\bm x = (x_1, \ldots, x_H)$.
Let $N$ be a Poisson point process on $\Omega := (0, +\infty)^H \times (0, +\infty) \times \Theta$ such that
\[
    N = \sum_{k \geq 1} \delta_{\mathbf{m}_k, z_k, x_k}
\] 
with intensity
\begin{equation}\label{eq:pois_intensity}
    \lambda_N(\dd \bm m \dd z \dd x) =\prod_{h=1}^H f(m_h) \dd m_h \nu^*(\dd z) \alpha(\dd x).
\end{equation}
Then, the collection of random measures $\mu^*_1, \ldots, \mu^*_H$ defined aw
\begin{equation}\label{eq:corm_poi}
    \mu^*_h(A) = \int_{\Omega} \pi_h(\bm{m}) z \indicator[x \in A] N(\dd \bm{m} \dd z \dd x)
\end{equation}
for all measurable $A$ is a compound random measures
\end{lemma}
\begin{proof}
The proof easily follows by writing explicitly \eqref{eq:corm_poi} as
\[
	\mu^*_h(A) = \sum_{k \geq 1} m_{hk} J_k \delta_{x}(A),
\]
observing that the points $(J_k, x_k)$ form a Poisson point process with intensity $\nu^*(\dd x) \alpha(\dd x)$.
Finally, from \eqref{eq:pois_intensity} it is clear that $m_{hk} \iid f$.
\end{proof}

We are now ready to prove Theorem \ref{teo:expectation}
\begin{proof}
    Write 
    \[
        \E[\ptilde_j(A)] = \E\left[\frac{\mutilde_j(A)}{\mutilde_j(\X)}\right] = \int_{\R_+} \sum_{h=1}^H \E \left[\lambda_{jh} e^{- u \sum_{k=1}^H \lambda_{jk} \mu^*_k (\X)} \mu^*_h(A)\right] \dd u
    \]
    where the second equality follows from writing $\mutilde_j(\cdot) = \sum_{h} \lambda_{jh} \mu^*(\cdot)$, the equality $t^{-1} = \int_{\R_+} e^{-ut} \dd u$ and an application of Fubini's theorem.
    By the tower property of the expected value, we further have
    \[
        \E[\ptilde_j(A)] = \int_{\R_+} \sum_{h=1}^H \E \left[ \lambda_{jh}  \E \left[ e^{- u \sum_{k=1}^H \lambda_{jk} \mu^*_k (\X)} \mu^*_h(A) \mid \Lambda \right] \right].
    \]
    Let us consider the inner expected value.
    Using \eqref{eq:corm_poi} we can write 
    \[
        \E \left[ e^{- u \sum_{k=1}^H \lambda_{jk} \mu^*_k (\X)} \mu^*_h(A) \mid \Lambda \right] = \E\left [\int_{\Omega} g(\bm m, z, x, N) N(\dd \bm m \dd z \dd x)\right]
    \]
    where 
    \[
        g(\bm m, z, x, N) = e^{- u \sum_{k=1}^H \lambda_{jk} \mu^*_k (\X)} \pi_h(\bm m) z \indicator[x \in A].
    \]
    Observe further that, although not explicitly written, $ \mu^*_k (\X)$ is of course a function of $N$.
    By the Campbell-Little-Mecke formula, 
    \[
        E \left[ e^{- u \sum_{k=1}^H \lambda_{jk} \mu^*_k (\X)} \mu^*_h(A) \mid \Lambda \right] = \int_{\Omega} g(\bm m, z, x, N + \delta_{(\mathbf m, z, x)})\lambda_N(\dd \bm m \dd z \dd x)
    \]
    where $\lambda_N$ is as in \eqref{eq:corm_poi}.
    Focusing on the integrand, we have
    \[
        g(\bm m, z, x, N + \delta_{(\mathbf m, z, x)}) = e^{- u \sum_{k=1}^H \lambda_{jk} (\mu^*_k + \pi_k(\mathbf{m}) z \delta_x)(\X) } \pi_h(\bm m) z \indicator[x \in A].
    \]
    With an abuse of notation, let us denote with $f$ the probability density of the $m_{hk}'s$, so that 
    \begin{align*}
        &E \left[ e^{- u \sum_{k=1}^H \lambda_{jk} \mu^*_k (\X)} \mu^*_h(A) \mid \Lambda \right] \\
        & \qquad = \int_\Omega \E \left[e^{- u \sum_{k=1}^H \lambda_{jk} \mu^*_k (\X) } \mid \Lambda \right] \prod_{k=1}^H e^{- u \lambda_{jk} m_k z }  m_h z \indicator[x \in A] \prod_{k=1}^H f(m_k) \dd m_k \nu^*(\dd z) \alpha(\dd x) \\
        & \qquad = \alpha(A) \E \left[e^{- u \sum_{k=1}^H \lambda_{jk} \mu^*_k (\X) } \mid \Lambda \right] \int_{\R_+} z \prod_{k \neq h} \int_{\R_+} e^{- u \lambda_{jk} m_k z } f(m_k) \dd m_k \\
        & \qquad \qquad \qquad \times  \int_{\R_+} e^{- u \lambda_{jh} m_h z } m_h f(m_h) \dd m_h \nu^*(\dd z)  \\
        & \qquad = \alpha(A) \E \left[e^{- u \sum_{k=1}^H \lambda_{jk} \mu^*_k (\X) } \mid \Lambda \right]\int_{\R_+} z \prod_{k \neq h} \mathcal{L}(u \lambda_{jk} z) \kappa(u \lambda_{jh} z, 1) \nu^*(\dd z) \\
        & \qquad = \alpha(A) \psi_{\rho}(u\lambda_{j1}, \ldots, u \lambda_{jH}) \int_{\R_+} z \prod_{k \neq h} \mathcal{L}_f(u \lambda_{jk} z) \kappa_f(u \lambda_{jh} z, 1) \nu^*(\dd z)
    \end{align*}
    where $\psi_{\rho}$ is the Laplace transform of $(\mu^*_1, \ldots, \mu^*_H)$ evaluated at the constant functions $u\lambda_{j1}, \ldots, u \lambda_{jH}$, $\mathcal{L}_f$ denotes the Laplace transform of the density $f$ and $\kappa_f(x, n) := \int e^{-x} m^n f(m) \dd m$.

    Hence, 
    \begin{align*}
        \E[\ptilde_j(A)] = \alpha(A) \sum_{h=1}^H \int \E\left[ \lambda_{jh}  \psi_{\rho}(u\lambda_{j1}, \ldots, u \lambda_{jH}) \int_{\R_+} z \prod_{k \neq h} \mathcal{L}_f(u \lambda_{jk} z) \kappa_f(u \lambda_{jh} z, 1) \nu^*(\dd z) \right] \dd u 
    \end{align*}
\end{proof}

\subsection{Proof of Proposition \ref{prop:cov_mu}}

\begin{align*}
    \Cov\left[\mutilde_j(A), \mutilde_\ell(B)\right] & = \Cov \left[ \sum_{h=1}^H \lambda_{j, h} \mu^*_h(A),  \sum_{k=1}^H \lambda_{\ell, k} \mu^*_k(B) \right]\\
    &= \E \left[ \sum_{h, k} \left( \lambda_{jh} \mu^*_h(A) - \bar{\lambda}_{jh} m^*_h(A) \right) \left( \lambda_{\ell k} \mu^*_k(B) - \bar{\lambda}_{\ell k} m^*_k(B) \right) \right] \\
    &= \sum_{h, k} \E \left[  \lambda_{jh} \lambda_{\ell_k}  \mu^*_h(A) \mu^*_k(B) \right] - \E[\lambda_{jh} \mu^*_h(A)] \bar{\lambda}_{\ell k} m^*_k(B) + \\
    & \qquad  \qquad - \bar{\lambda}_{jh} m^*_k(A) \E[\lambda_{\ell k} \mu^*_k(B)] + \bar{\lambda}_{jh} \bar{\lambda}_{\ell k}  m^*_k(A) m^*_k(B) \\
    &= \sum_{h, k} \E[\lambda_{jh} \lambda_{\ell k}] \E [\mu^*_h(A) \mu^*_k(B) ] - \bar{\lambda}_{jh} \bar{\lambda}_{\ell k}  m^*_k(A) m^*_k(B) 
     \end{align*}
    In the most general case, we thus have that
    \begin{align*}
    \Cov\left[\mutilde_j(A), \mutilde_\ell(B)\right]  &= \sum_{h} \E[\lambda_{jh} \lambda_{\ell h}] \E [\mu^*_h(A) \mu^*_h(B) ] - \bar{\lambda}_{jh} \bar{\lambda}_{\ell h}  m^*_k(A) m^*_h(B) + \\
    & \qquad \qquad \sum_{h \neq k} \E[\lambda_{jh} \lambda_{\ell k}] \E [\mu^*_h(A) \mu^*_k(B) ] - \bar{\lambda}_{jh} \bar{\lambda}_{\ell k}  m^*_k(A) m^*_k(B) \\ 
    &= \sum_{h} \E[\lambda_{jh} \lambda_{\ell h}] \Cov (\mu^*_h(A), \mu^*_h(B)) + \Cov(\lambda_{jh}, \lambda_{\ell h}) m^*_h(A) m^*_h(B) + \\
    & \qquad \qquad  \sum_{h \neq k} \E[\lambda_{jh} \lambda_{\ell k}] \Cov (\mu^*_h(A), \mu^*_k(B)) + \Cov(\lambda_{jh}, \lambda_{\ell k}) m^*_h(A) m^*_k(B) 
    \end{align*}

\subsection{Covariances and Correlations}

\paragraph{The case of $\Ga(\Psi, 1)$ scores.}

Specializing Proposition \ref{prop:cov_mu} we have
\begin{align*}
	\Cov\left[\mutilde_j(A), \mutilde_\ell(A)\right] &= \E[\mu^*_1(A)^2] H \psi^2 + (c_A + m_A^2) H(H-1) \psi^2 - m_A^2 H^2 \psi^2 \\
	&= (\Var[\mu^*_1(A)] H + c_A H(H-1)) \psi^2 
\end{align*}
Moreover,
\begin{align*}
	\Var[\mutilde_j(A)] &= \E[\mu^*_1(A)^2] H \psi(\psi+1) + (c_A + m_A^2) H(H-1) \psi^2 - m_A^2 H^2 \psi^2 \\
	&= (\Var[\mu^*_1(A)] H + c_A H(H-1)) \psi^2 +  \E[\mu^*_1(A)^2] H \psi
\end{align*}
Simple algebra leads to Equation \eqref{eq:corr_iid_scores}

\paragraph{The multiplicative gamma process case.}
Using standard properties of inverse-gamma distributed random variables, we get
\begin{align*}
	& \Cov\left[\mutilde_j(A), \mutilde_\ell(A)\right] =  \\
	& \qquad \E[\mu^*_1(A)^2] \left(\sum_{h=1}^H (a_2 - 1)^{-h + 1} (a_2 - 2)^{-h + 1} \right) (a_1 - 1)^{-1} (a_1 - 2)^{-1} \left(\frac{\nu}{\nu - 2}\right)^2 \\
	& \qquad + (c_A + m_A^2) \left( 2 \sum_{h < k} (a_2 - 1)^{-k + 1} (a_2 - 2)^{-h + 1} \right)(a_1 - 1)^{-1} (a_1 - 2)^{-1} \left(\frac{\nu}{\nu - 2}\right)^2  \\
	& \qquad - m_A^2 \left(\sum_{h, k} (a_2 -1)^{-h - k + 1} \right) (a_1 - 1)^{-2}  \left(\frac{\nu}{\nu - 2}\right)^2
\end{align*}
and
\begin{align*}
	& \Var[\mutilde_j(A)] = \\
	& \qquad \E[\mu^*_1(A)^2] \left(\sum_{h=1}^H (a_2 - 1)^{-h + 1} (a_2 - 2)^{-h + 1} \right) (a_1 - 1)^{-1} (a_1 - 2)^{-1} \frac{\nu^2}{(\nu - 2)(\nu - 4)} \\
	& \qquad + (c_A + m_A^2) \left( 2 \sum_{h < k} (a_2 - 1)^{-k + 1} (a_2 - 2)^{-h + 1} \right)(a_1 - 1)^{-1} (a_1 - 2)^{-1} \left(\frac{\nu}{\nu - 2}\right)^2  \\
	& \qquad - m_A^2 \left(\sum_{h, k} (a_2 -1)^{-h - k + 1} \right) (a_1 - 1)^{-2}  \left(\frac{\nu}{\nu - 2}\right)^2
\end{align*}
Note that the only term differing in the expressions of $\Cov\left[\mutilde_j(A), \mutilde_\ell(A)\right]$ and $ \Var[\mutilde_j(A)]$  is the last factor in the first row.

\subsection{Proof of Propsition \ref{prop:moments_corm}}

The first point follows directly from the definition of the gamma process. Regarding the second one, we recall a general expression given in \cite{griffin2017compound}.

\begin{theorem}{[\emph{Mixed moments of CoRMs}, \citep[Theorem 6,][]{griffin2017compound}]}
Let $q_h \geq 0$, $i = h, \ldots, H$ such that $\sum_h q_h = k$. Then
\begin{multline}
	\E\left[\prod_{h=1}^H \left( \mu^*_h(A)^{q_h}\right)\right] = \prod_{h} q_h! \left(\sum_{j=1}^{k} \alpha(A)^\ell \right) \\
	 \times \sum_{j=1}^k \ \sum_{\bm \eta, \bm s_1, \ldots, \bm s_j \in p_j(k) } \prod_{i=1}^j  \frac{1}{\eta_i !} \left[ \prod_{h=1}^H \frac{(\phi)_{s_{hi}}}{s_{hi}!} \int z^{s_{1i}+ \cdots+ s_{Hi}} \nu^*(\dd z)   \right]^{\eta_j}
\end{multline}
where $p_j(k)$ is the set of vectors $(\bm \eta, \bm s_1, \ldots, \bm s_j)$, $\eta = (\eta_1, \ldots, \eta_j)$, $\bm s_i = (s_{i1}, \ldots, s_{iH})$, such that $\eta_i$ is positive, $\sum \eta_i = k$, $\bm 0 \prec \bm s_1 \prec \cdots \prec \bm s_j$ and $\sum_{i=1}^j \eta_i(s_{i1} + \cdots + c_{Hi}) = k$.
\end{theorem}

It suffices to consider the case $\bm q = (1, 1, 0, \ldots, 0)$. Then, the problem consists in understanding how the sets $p_j(2)$ are made for $j=1, 2$.
The only possible vector $\bm \eta$ in $p_1(2)$ is $\bm \eta = (2)$. Therefore the only possible $\bm s_1$ is $\bm s_1 = (1, 0, \ldots, 0)$.
Hence the sum over $\bm \eta, \bm s_1, \ldots, \bm s_j \in p_j(k)$ when $j=1$ equals to
\[
	\frac{1}{2} \left[ \phi \int z \nu^*(\dd z) \right]^2
\]
When $j=2$, we have that the possible  $\bm \eta$'s are $(0, 2)$, $(1, 1)$, $(2, 0)$. Note that the first and last candidate cannot satisfy $\sum_{i=1}^j \eta_i(s_{i1} + \cdots + c_{Hi}) = k$ for any choice of $\bm s$. Therefore, we can consider $\bm \eta = (1, 1)$, leading to $\bm s_1 = (0, 1, 0, \ldots, 0)$ and $\bm s_2=(1, 0, \ldots, 0)$. 
Hence the sum over $\bm \eta, \bm s_1, \ldots, \bm s_j \in p_j(k)$ when $j=2$ equals to
\[
	\phi^2 \left[  \int z \nu^*(\dd z) \right] \left[ \int z \nu^*(\dd z)\right]
\]
Finally, observe that when the CoRM has gamma marginals, $\int z \nu^*(\dd z) = B(1, \phi)$, where $B$ is the Beta function.
This concludes the proof.

\subsection{Proof of Proposition \ref{prop:lie_proj}}

Let $\{E_n\}_n$ be the generators for $\mathfrak{sl}(H)$. Then
\[
	\Pi_{\mathfrak{sl}(H)} = \sum_n \mbox{tr}(X E_n)E_n
\]
It is easy to see that such a set of generators is given by:
\[
	\left\{ \bigcup_{\ell \neq m} A: A_{ij} = \delta_{\ell, m}(i, j) \right\} \cup \left\{ \bigcup_{\ell=1}^{H-1} A: A_{i, i} = 1, A_{i+1, i+1} = -1 \right\}
\]
which consists of $H(H-1)$ (first term) plus $H-1$ (second term) elements. We call the two sets above $A^*_1$ and $A^*_2$ respectively.

For numerical purposes, we don't need to compute the inner product and sum with all the $H^2 - 1$ elements in the basis. In fact note that when $E_n \in A^*_1$, say $E_n$ is nonzero only in element $i,j$, $\mbox{tr} (X E_n) E_n$ is a matrix whose only nonzero entry is the $j,i$-th with value $X_{i, j}$. Therefore
\[
	\sum_{E_n \in A^*_1} \mbox{tr}(X E_n)E_n = (X - \mbox{diag}(X))^T,
\]
where $\mbox{diag}(X)$ is the diagonal matrix with entries equal to the diagonal of $X$.
Similarly, when $E_n \in A^*_2$,  $\mbox{tr} (X E_n) E_n$ is a diagonal matrix whose nonzero entries are the $(i, i)$-th and $(i+1, i+1)$-th and are equal to $\pm X_{i, i} - X_{i+1, i+1}$ respectively.

\section{Slice Sampling Algorithm}\label{sec:app_slice}

Let $T_j = \sum_{\ell \geq 1} (\Lambda M)_{j \ell} J_\ell$ and introduce auxiliary cluster allocation variables $c_{j, i}$ (one for each observation $y_{j, i})$ as well as auxiliary latent variables $U_j$ such that $U_j \mid T_j \sim \Ga(n_j, T_j)$. Standard computations lead to the extended likelihood
\begin{multline*}
    p(\{y_{j, i}\}, \{c_{j, i}\}, \{u_j\} \mid \cdots) =  \left[\prod_{j=1}^g \frac{1}{\Gamma(n_j)} u_j^{n_j - 1} \right] \times \\ 
    \prod_{j=1}^g \prod_{i=1}^{n_j} f(y_{j, i} \mid \theta_{c_{j, i}}) (\Lambda M)_{j, c_{j, i}} J_{c_{j, i}} \times 
    \exp \left( - \sum_{j=1}^g u_j \sum_{\ell=1}^\infty (\Lambda M)_{j, \ell} J_\ell \right) 
\end{multline*}
We further introduces auxiliary slice variables $s_{j, i}$ so that 
\begin{multline*}
    p(\{y_{j, i}\}, \{c_{j, i}\}, \{u_j\} \mid \cdots) =  \left[\prod_{j=1}^g \frac{1}{\Gamma(n_j)} u_j^{n_j - 1} \right] \times \\
    \prod_{j=1}^g \prod_{i=1}^{n_j} f(y_{j, i} \mid \theta_{c_{j, i}}) (\Lambda M)_{j, c_{j, i}} I(s_{j, i} < J_{c_{j, i}}) \times 
    \exp \left( - \sum_{j=1}^g u_j \sum_{\ell=1}^\infty (\Lambda M)_{j, \ell} J_\ell \right)
\end{multline*}
where $I(\cdot)$ denotes the indicator function.
Then, we can devide between \emph{active} and \emph{non-active} components: let $L = \min s_{j, i}$, $J^{a} = \{J_\ell \text{ s.t. } J_\ell > L\}$ and $J^{na} = J \setminus J^a$, we further denote with $k$ the cardinality of $J^a$, observe that $k$ is finite almost suerly. Analogously define $M^a$ the $H \times k$ matrix with columns $\{m_\ell \text{ s.t. } J_\ell > L\}$ and $M^{na}$ in a similar fashion.
The likelihood can be rewritten as
\begin{align*}
    & p(\{y_{j, i}\}, \{c_{j, i}\}, \{u_j\} \mid \cdots) =  \left[\prod_{j=1}^g \frac{1}{\Gamma(n_j)} u_j^{n_j - 1} \right] \times \\
    & \qquad \prod_{j=1}^g \prod_{i=1}^{n_j} f(y_{j, i} \mid \theta_{c_{j, i}}) (\Lambda M)_{j, c_{j, i}} I(s_{j, i} < J_{c_{j, i}}) \times 
    \exp \left( - \sum_{j=1}^g u_j \sum_{\ell=1}^k (\Lambda M^{a})_{j, \ell} J^{a}_\ell \right) \\
    & \qquad \exp \left( - \sum_{j=1}^g u_j \sum_{\ell=1}^\infty (\Lambda M^{na})_{j, \ell} J^{na}_\ell \right)
\end{align*}

To compute posterior inference, we need to be able to marginalize over $M^{na}$ and $J^{na}$, and compute
\begin{equation}\label{eq:lap_fac}
    \E \left[ \exp \left( - \sum_{j=1}^g u_j \sum_{\ell=1}^\infty (\Lambda M^{na})_{j, \ell} J^{na}_\ell \right) \Big | \Lambda \right]
\end{equation}
We manipulate the sum in the exponential to get 
\begin{align*}
    \exp \left( - \sum_{j=1}^g u_j \sum_{\ell=1}^\infty (\Lambda M^{na})_{j, \ell} J^{na}_\ell \right) & = \exp \left( - \sum_{j=1}^g u_j \sum_{\ell=1}^\infty \sum_{h=1}^H \lambda_{j, h} m^{na}_{h, \ell} J^{na}_\ell \right) \\
    &= \exp \left( - \sum_{h=1}^H \sum_{j=1}^g u_j \lambda_{j, h} \sum_{\ell = 1}^\infty m^{na}_{h, \ell} J^{(na)}_\ell  \right) \\
    &= \exp \left( - \sum_{h=1}^H \left[ \left(\sum_{j=1}^g u_j \lambda_{j, h} \right) \left( \sum_{\ell = 1}^\infty m^{na}_{h, \ell} J^{(na)}_\ell \right) \right]  \right)
\end{align*}
So that \eqref{eq:lap_fac} can be computed by virtue of Theorem 1 in \cite{griffin2017compound}, replacing $v_j$ (in their notation) with $\sum_{j=1}^g u_j \lambda_{j, h} $.

Then, the MCMC algorithm follows the same lines of the slice sampling algorithm in \cite{griffin2017compound}.

\section{Aligning densities in higher-dimension}\label{sec:align_highdim}

Computing the $L_2$ distance between functions is easy when the dimension of the data space is small, which is always the case in our simulations.
In higher dimensional settings, we suggest instead the following dissimilarity function
\[
	d(\hat \mu, \mu^\prime)^2 = \inf_{T \in \Gamma(\hat \mu, \mu^\prime)} \sum_{h, k=1}^K W^2_2(f(\cdot \mid \hat \theta^*_h), f(\cdot \mid \theta^\prime_k))  T_{hk}
\]
where $\{\hat \theta^*_h\}_h$ and $\{\theta^\prime_k\}_k$ are the atoms in $\hat \mu$ and $\mu^\prime$ respectively, $\Gamma(\hat \mu, \mu^\prime)$ denotes all the $K \times K$ matrices whose row-sums are equal to the normalized weights in $\hat \mu$ and the column-sums are equal to the normalized weights in $\mu^\prime$. 

That is, the distance corresponds to the Wasserstein distance between two atomic probability measures. 
The associated ground cost is $W^2_2(f(\cdot \mid \hat \theta^*_h), f(\cdot \mid \theta^\prime_k))$ that is the squared Wasserstein distance between the probability measure with density $f(\cdot \mid \hat \theta^*_h)$ and the one with density $f(\cdot \mid \theta^\prime_k)$. 
This choice of ground cost ensures that the specific choice of the kernel density $f$ is taken into account.

In particular, $W^2_2(f(\cdot \mid \hat \theta^*_h), f(\cdot \mid \theta^\prime_k))$ can be easily computed for location-scatter families of probability densities. 
Let $\mathcal L$ denote a generic law of a random variable, and $X_0$ a $d$-dimensional random vector with law $P_0$ such that $\E[\|X\|^2] < +\infty$ and $P_0$ is absolutely continuous with respect to the $d$-dimensional Lebesgue measure. Then a location-scatter family is the set of random variables
\[
	\{\mathcal L(\Sigma^{1/2} X_0 + \mu), \text{ such that } \Sigma  \text{ is symmetric and positive definite, } \mu \in \mathbb R^d \}
\]
This definition obviously encompasses the popular Gaussian density but also the Student-$t$, Laplace, and discrete and continuous uniform distributions among others.

Let $f(\cdot \mid \mu_i, \Sigma_i)$, $i=1, 2$ denote the densities of two random variables in the location-scatter family under consideration. Theorem 2.1 and Corollary 3.12 in \cite{alvarez_esteban_wass} entail that 
\[
	W_2^2\left( f(\cdot \mid \mu_1, \Sigma_1), f(\cdot \mid \mu_2, \Sigma_2)\right) = \|\mu_1 - \mu_2\|^2 + \mbox{trace}\left( \Sigma_1 + \Sigma_2 - 2(\Sigma_1^{1/2} \Sigma_2 \Sigma_1^{1/2})^{1/2} \right).
\]
Hence, the proposed distance can be computed exactly. 
The main computational bottleneck is the computation of the matrix square root. Its exact computation requires computing the eigendecomposition of the matrix, whose computational cost scales cubically with the dimension. Otherwise, several approximate iterative algorithms have been proposed.

\section{Additional Simulations and Plots}\label{sec:app_simu}

Figure \ref{fig:corr_lgmrf} shows the correlation between $\mutilde_j(A)$ and $\mutilde_\ell(A)$ under  prior \eqref{eq:lambda_gmrf} above. We consider a simple setting with three areas $i, j, \ell$ such that areas $i$ and $j$ are neighboring while area $\ell$ is not connected to either $i$ and $j$.

Figure \ref{fig:lambda_var_prior1} shows the variance of the ratio $r_{j\ell}^k$ defined in Equation \eqref{eq:jump_ratio} under different priors for $\Lambda$. As expected, the variance quickly drops to zero when the $\lambda_{jh}$'s are i.i.d. as $H$ increases. The same happens when we assume that $\Lambda$ follows a shrinkage prior, but the decay is slower.

Figure \ref{fig:lambda_var_prior2} shows the effect of the a priori variance of the $\lambda_{jh}$'s on the variance of $r_{j\ell}^k$.

Figure \ref{fig:invalsi_dendrogram} shows the dendrogram of the hierarchical clustering on the rows of $\Lambda^\prime$ on the Invalsi dataset.

Figure \ref{fig:income_eda} shows some exploratory data analysis for the US income dataset analyzed in Section \ref{sec:income_data}.

Figure \ref{fig:income_latent_est} shows the estimates of the latent factor measures for the US income dataset after the post-processing.

\begin{figure}
\centering
\includegraphics[width=\linewidth]{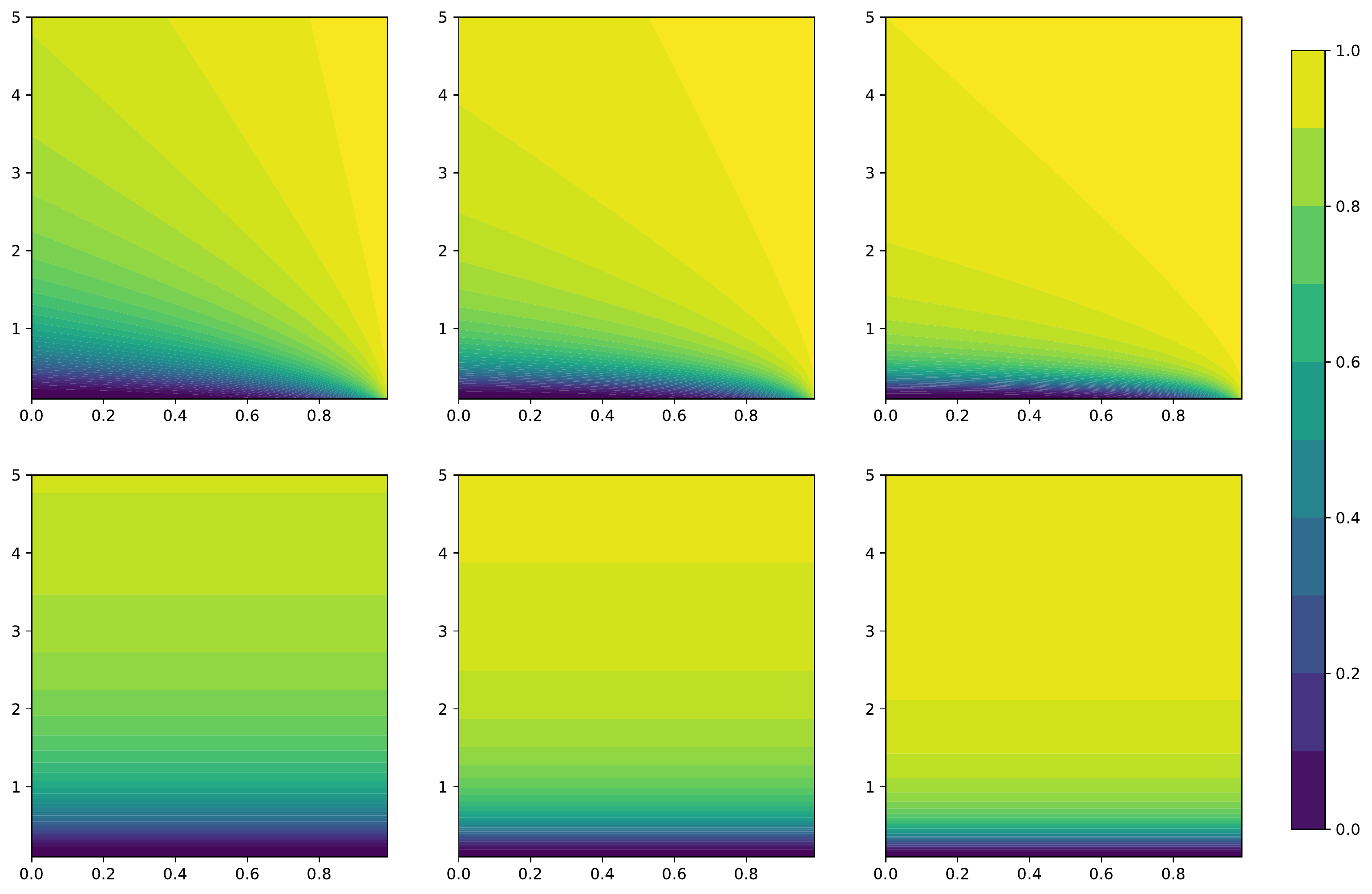}
\caption{Correlation between neighboring $\mutilde_i(A)$ and $\mutilde_j(A)$ (top row) and between disconnected $\mutilde_i(A)$ and $\mutilde_\ell(A)$ for a set $A$ such that $\alpha(A) = 0.5$ under prior \eqref{eq:lambda_gmrf}. From left to right $H=4, 8, 16$. The values of $\rho$ vary across the $x$-axis in each plot, the values of $\tau$ across the $y$-axis.}
\label{fig:corr_lgmrf}
\end{figure}

\begin{figure}
    \centering
    \begin{subfigure}{0.33\linewidth}
        \centering
        \includegraphics[width=\linewidth]{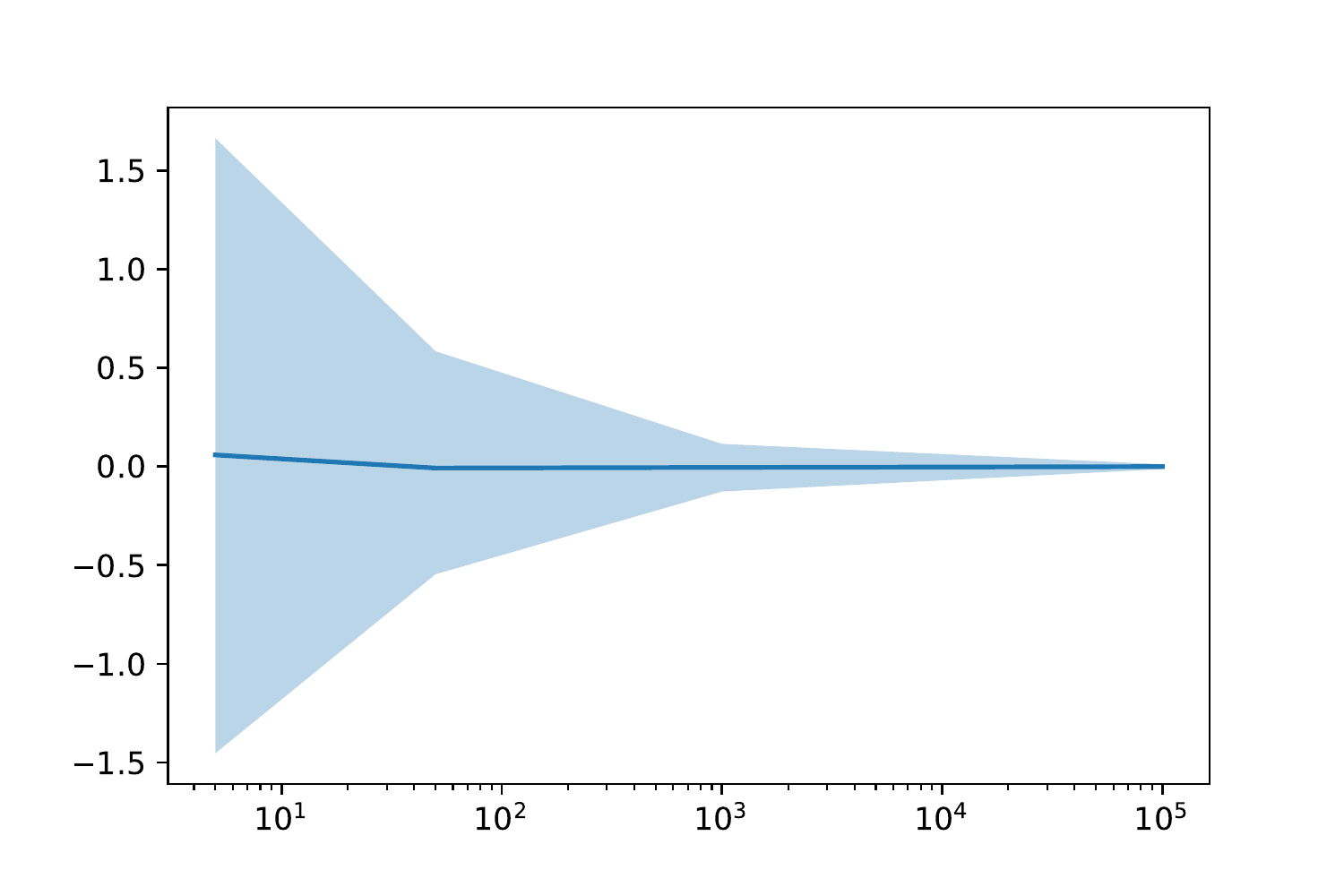}
    \end{subfigure}%
    \begin{subfigure}{0.33\linewidth}
        \centering
        \includegraphics[width=\linewidth]{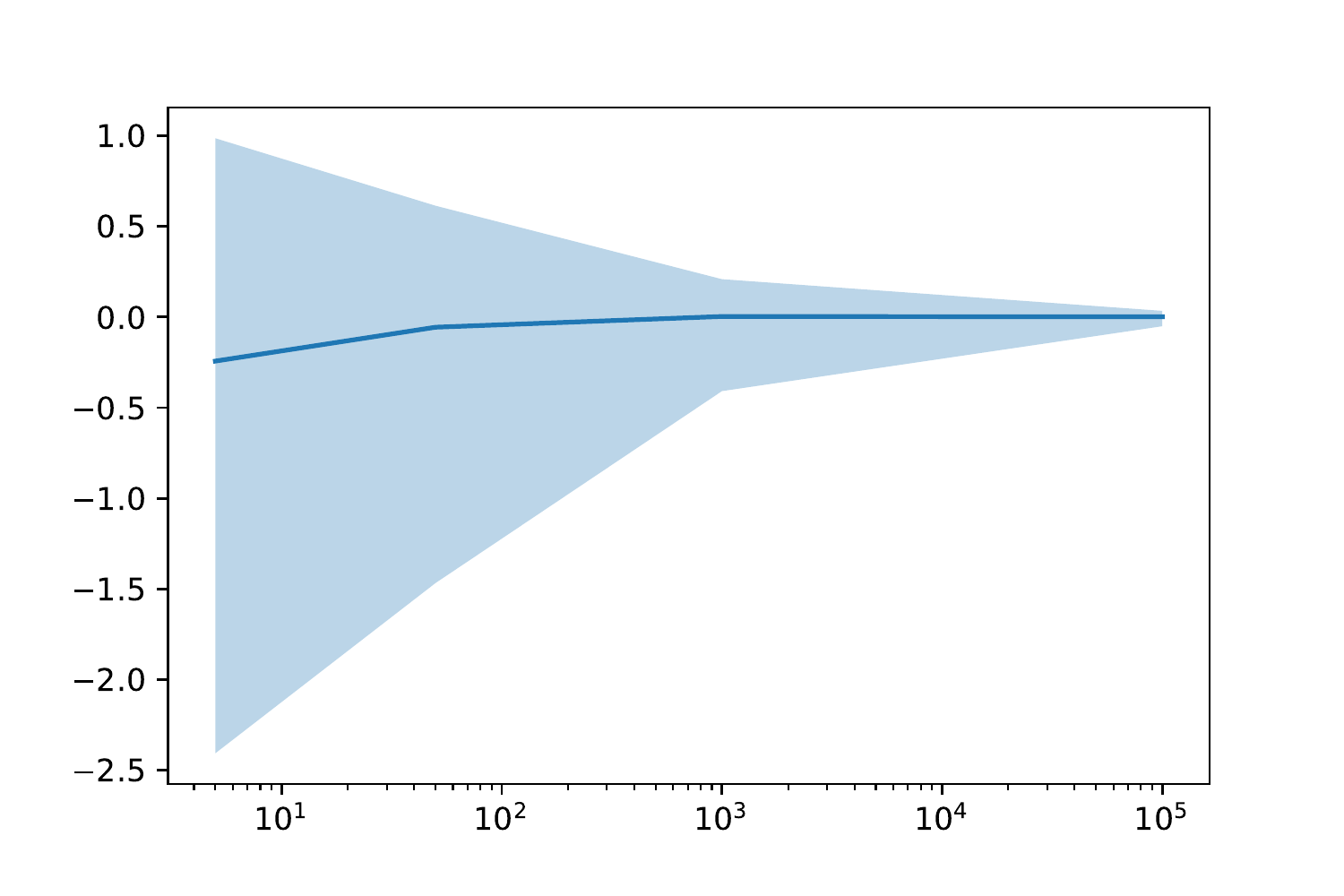}
    \end{subfigure}%
    \begin{subfigure}{0.33\linewidth}
        \centering
        \includegraphics[width=\linewidth]{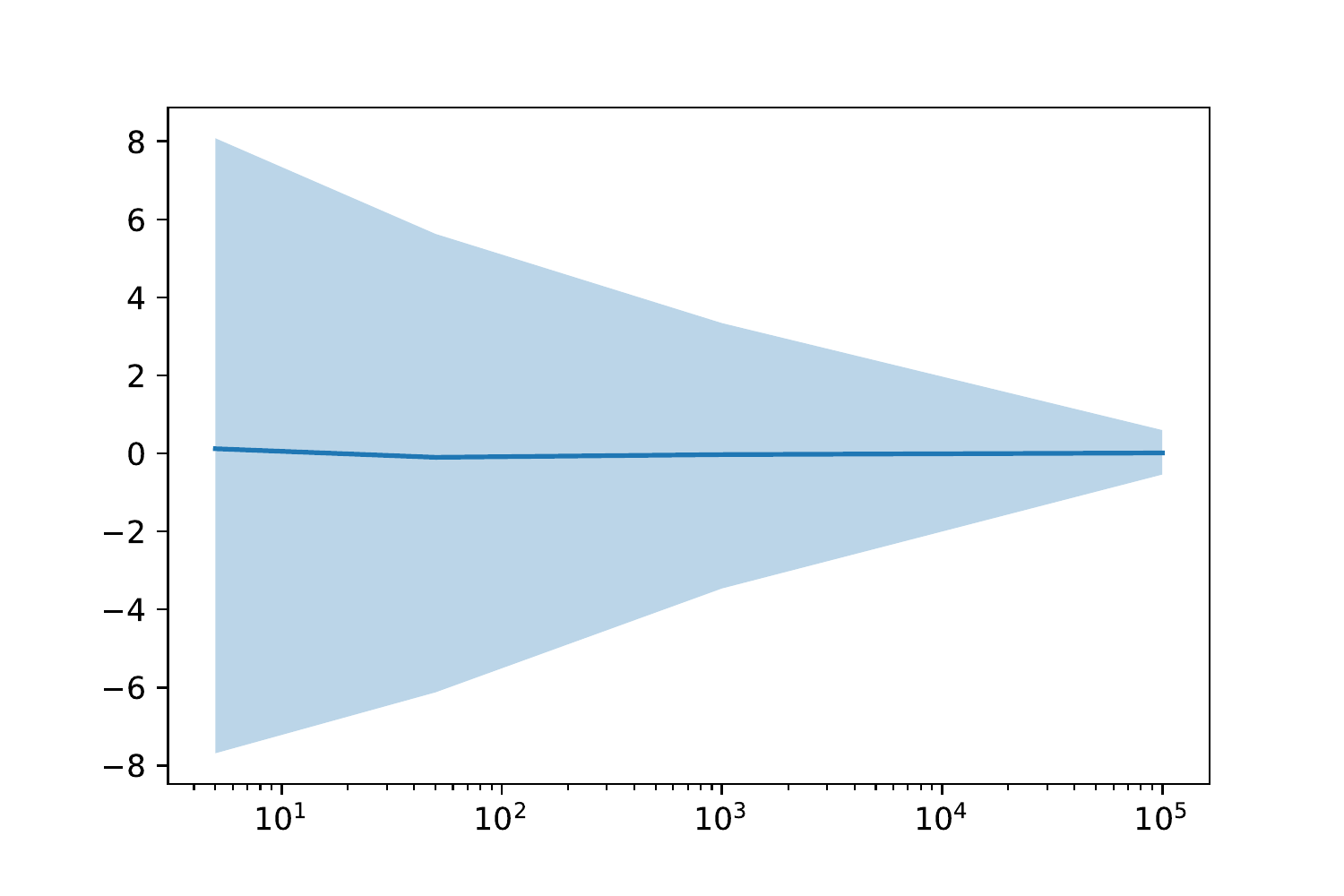}
    \end{subfigure}

    \caption{Monte Carlo estimate of $\log r_{j\ell}$ as a function of $H$ under different priors: from left to right, $\lambda_jh \iid  \mbox{Ga}(1, 1)$, $\bm \lambda_j = (\lambda_{j1}, \ldots, \lambda_{jH}) \iid \mbox{MGP}(2, 1)$, $\bm \lambda_j \iid \mbox{CUSP}$. The solid line represents the Monte Carlo average over $1,000$ simulations. The shaded area are $95\%$ confidence intervals.}
    \label{fig:lambda_var_prior1}
\end{figure}

\begin{figure}
    \centering
    \includegraphics[width=0.6\linewidth]{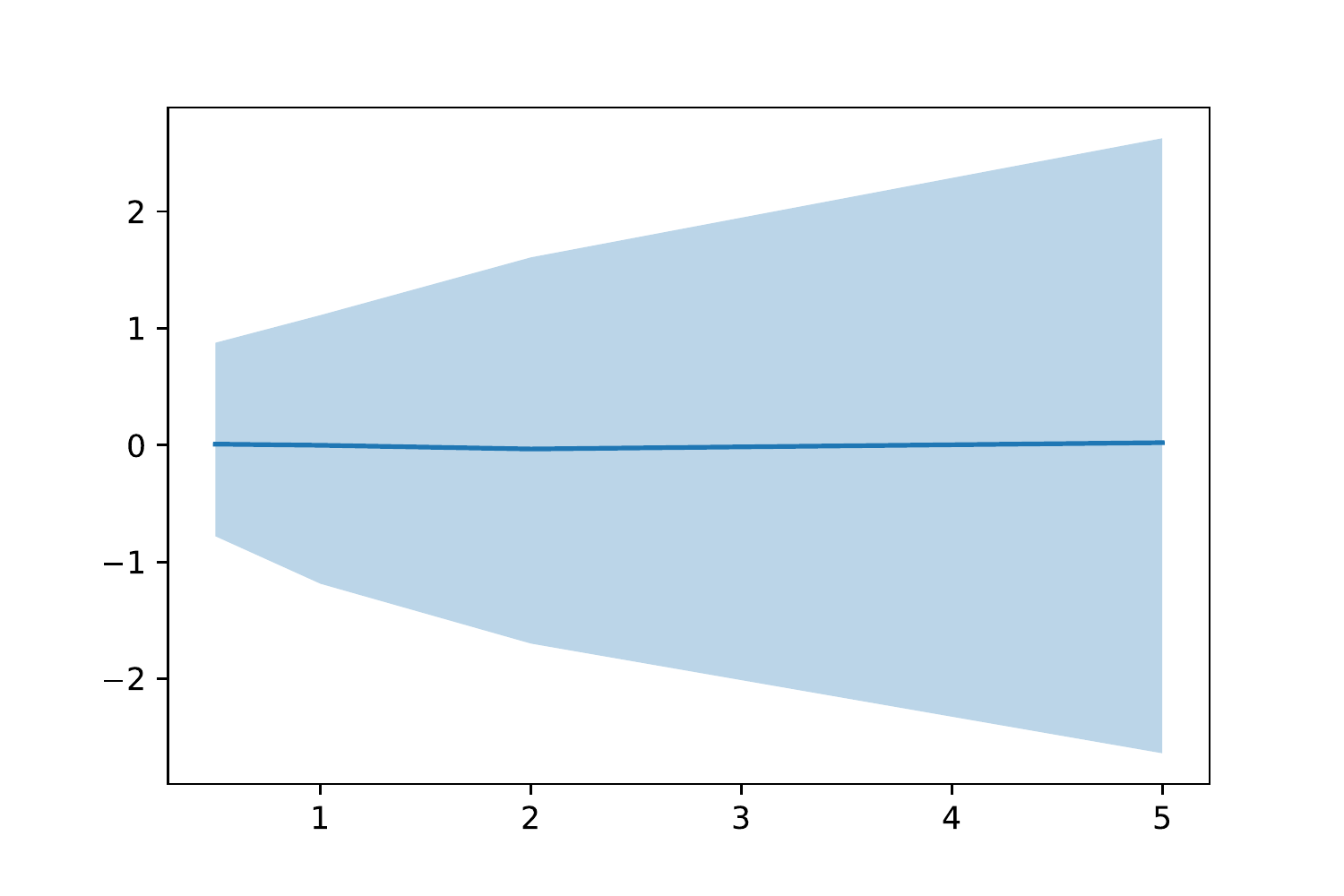}
    \caption{Monte Carlo estimate of $\log r_{j\ell}$ when $\lambda_{jh}$ are i.i.d gamma variables with mean equal to 1 and increasing variance (x-axis).}
    \label{fig:lambda_var_prior2}
\end{figure}

\begin{figure}
\centering
\includegraphics[width=0.6\linewidth]{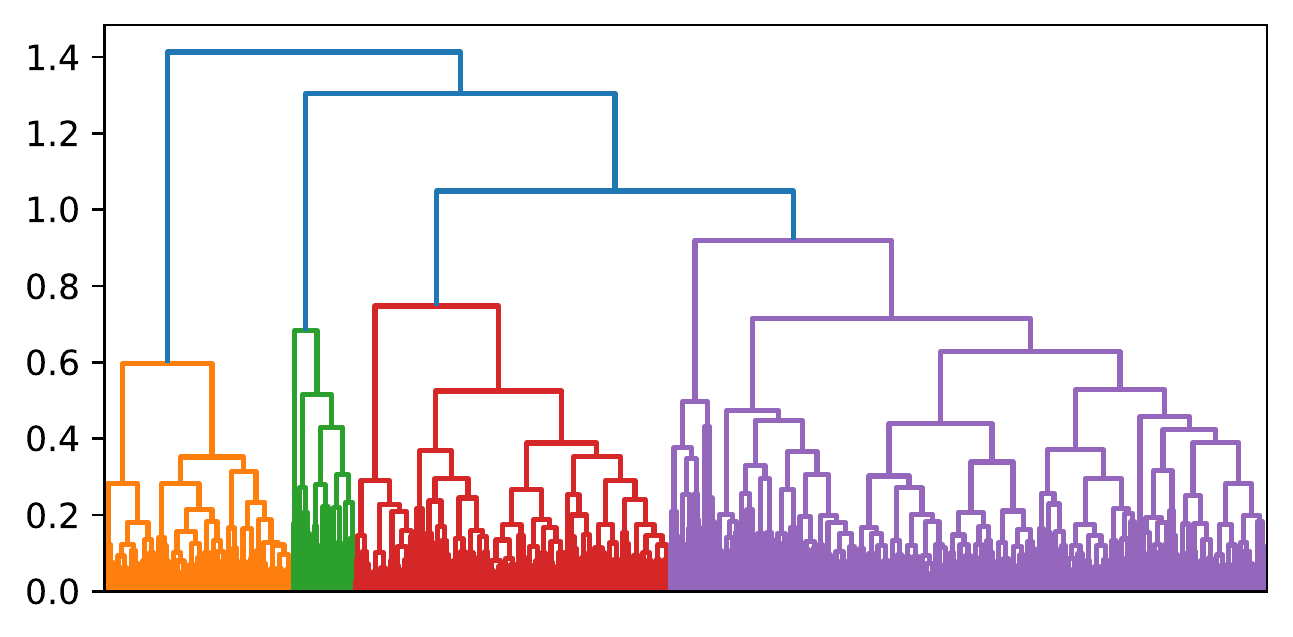}
\caption{Dendrogram for the hierarchical clustering with complete linkage on the rows of $\Lambda^\prime$ on the Invalsi dataset.}
\label{fig:invalsi_dendrogram}
\end{figure}

\begin{figure}
\centering
\includegraphics[width=\linewidth]{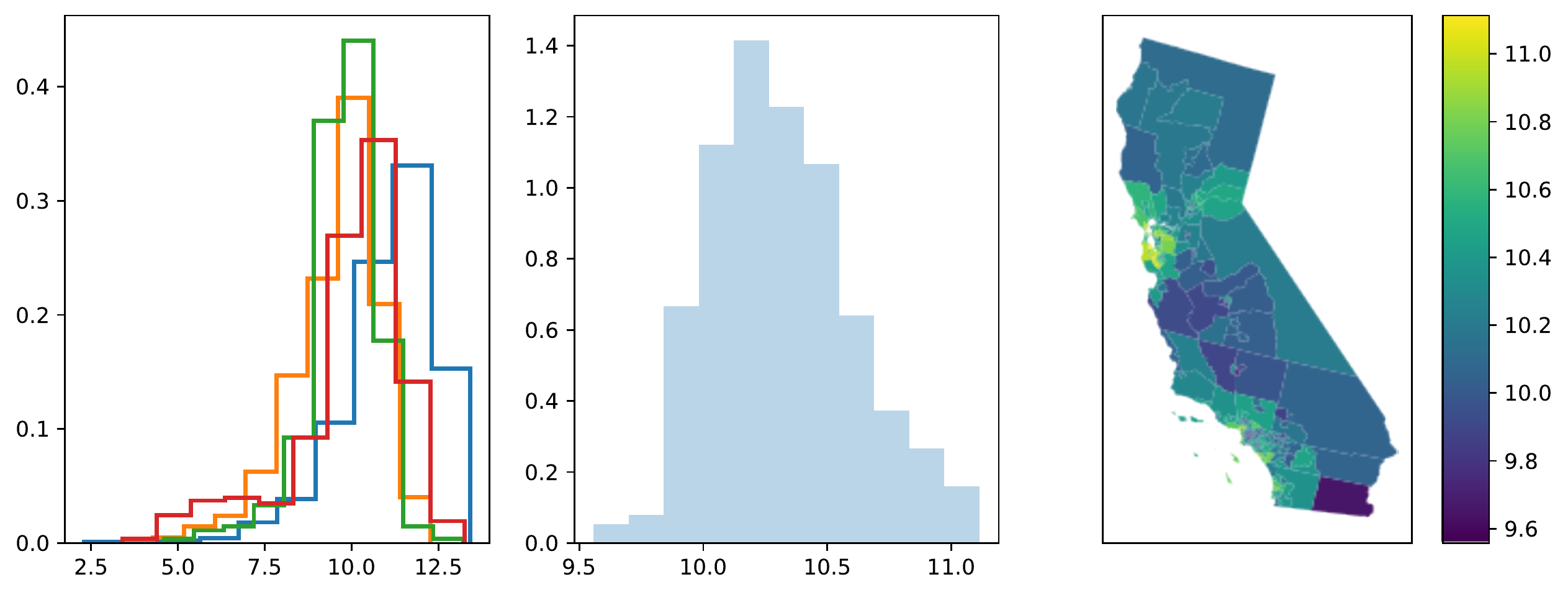}
\caption{From left to right: histogram of the (log) incomes in five randomly sampled PUMAs, histogram of the average (log) income across all the PUMAs, average (log) income displayed in a map}
\label{fig:income_eda}
\end{figure}

\begin{figure}
\centering
\includegraphics[width=\linewidth]{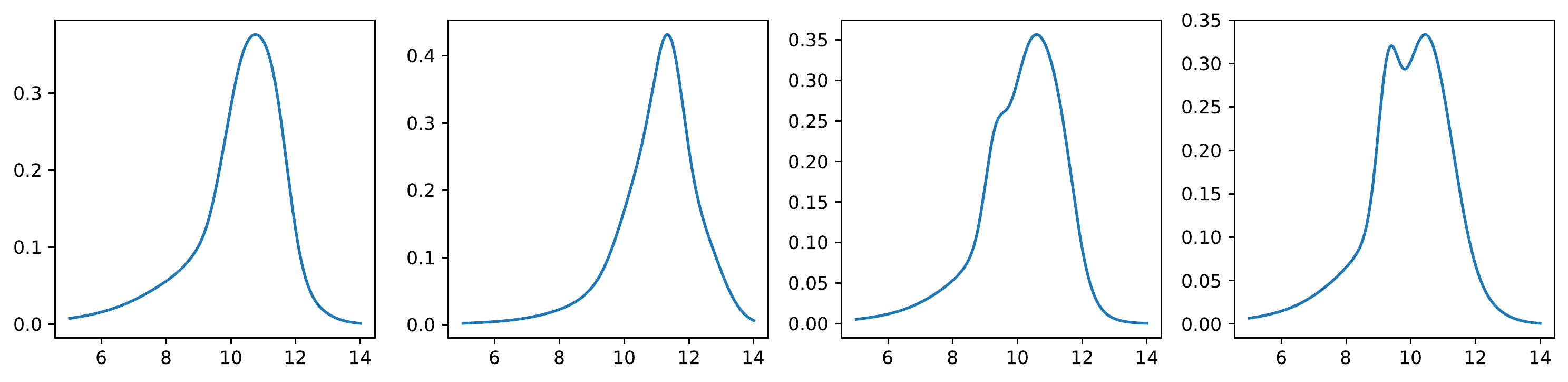}
\caption{Estimates of $\int_\Theta  f(y \mid \theta )\mu^\prime_h(\dd \theta) / \mu^\prime_h(\Theta)$  after post-processing in the US Income example.}
\label{fig:income_latent_est}
\end{figure}
    
\end{document}